\documentclass[journal]{IEEEtran}

\PassOptionsToPackage{hyphens}{url}

\usepackage{amsmath,amssymb,amsfonts,amsthm,mathrsfs}
\usepackage{algorithm}
\usepackage{algpseudocode}
\usepackage{adjustbox}
\usepackage{array}
\usepackage{booktabs}
\usepackage{cite}
\usepackage{graphicx}
\usepackage{makecell}
\usepackage{stfloats}
\usepackage{textcomp}
\usepackage{tabularx}
\usepackage{threeparttable}
\usepackage{xcolor}
\usepackage[hidelinks]{hyperref}
\usepackage{extarrows}

\hypersetup{breaklinks=true}

\definecolor{revisionblue}{RGB}{0,112,192}
\newif\ifmarkrevisions
\markrevisionstrue
\makeatletter
\def\rev@blue{0 0.43922 0.75294 rg 0 0.43922 0.75294 RG}
\def\rev@black{0 g 0 G}
\newcommand{\rev@set}[1]{%
  \let\current@color#1%
  \pdfcolorstack\@pdfcolorstack set{\current@color}%
}
\DeclareRobustCommand{\revstart}{\ifmarkrevisions\leavevmode\rev@set\rev@blue\fi}
\DeclareRobustCommand{\revend}{\ifmarkrevisions\rev@set\rev@black\fi}

\DeclareRobustCommand{\revcolor}{\ifmarkrevisions\rev@set\rev@blue\fi}
\makeatother

\def\beq{\begin{equation}}
\def\eeq{\end{equation}}
\def\bea{\begin{eqnarray}}
\def\eea{\end{eqnarray}}
\def\ba{\begin{array}}
\def\ea{\end{array}}

\def\bitem{\begin{itemize}}
\def\eitem{\end{itemize}}
\def\ben{\begin{enumerate}}
\def\een{\end{enumerate}}

\definecolor{bgrd}{rgb}{1,1,1}
\definecolor{gray}{rgb}{0.5,0.5,0.5}
\definecolor{dkr}{rgb}{0.7,0.1,0.2}
\definecolor{dkb}{rgb}{0.1,0.1,0.8}

\def\edoc{\end{document}}

\newlength{\eqtagspace}
\newcommand{\eqfitu}[1]{%
  \adjustbox{max width=\linewidth}{\(\displaystyle #1\)}%
}

\newcolumntype{Y}{>{\raggedright\arraybackslash}X}
\newcolumntype{Z}{>{\centering\arraybackslash}X}
\newcolumntype{L}[1]{>{\raggedright\arraybackslash}p{#1}}
\newcolumntype{C}[1]{>{\centering\arraybackslash}p{#1}}

\newtheorem{proposition}{Proposition}

\newtheorem{definition}{Definition}
\newtheorem{remark}{Remark}

\begin{document}

\title{Reachable-Set Decomposition for Real-Time Aggregation of Multi-Zone HVAC Fleets}

\author{Jingguan Liu,~Xiaomeng Ai,~Cong Chen,~Shaoze Li,~Shichang Cui,~Jiakun Fang,~and Jinyu Wen
\thanks{This work was supported by the National Natural Science Foundation of China under Grant U25B6017.~\textit{(Corresponding author:~Xiaomeng Ai.)}}%
\thanks{Jingguan Liu, Xiaomeng Ai, Shichang Cui, Jiakun Fang, and Jinyu Wen are with the State Key Laboratory of Advanced Electromagnetic Technology, Huazhong University of Science and Technology, Wuhan 430074, China (e-mail: \url{spencerplusmail@foxmail.com}; \url{xiaomengai@hust.edu.cn};\url{shichang_cui@hust.edu.cn}; \url{jfa@hust.edu.cn}; \url{jinyu.wen@hust.edu.cn}).}%
\thanks{Cong Chen and Shaoze Li are with the Thayer School of Engineering, Dartmouth College, Hanover, NH 03755, USA (e-mail: \url{Cong.Chen@dartmouth.edu}; \url{shaoze_li@163.com}).}
}

\markboth{under review}%
{Liu \MakeLowercase{\textit{et al.}}: Reachable-Set Decomposition for Real-Time Aggregation of Multi-Zone HVAC Fleets}

\maketitle

\begin{abstract}
Aggregating building heating, ventilation, and air-conditioning (HVAC) fleets can provide substantial real-time flexibility to power systems. However, practical deployment requires scalable characterization of multi-zone, multi-period flexibility and guaranteed feasible disaggregation as temperature states and exogenous inputs are revealed sequentially. This paper develops an offline–online reachable-set decomposition framework that combines offline temporal decomposition with a scalable online fleet coordination policy. Offline, backward reachable sets encode remaining-horizon feasibility as per-period state constraints, while tailored polytopic inner approximations provide tractable representations of these sets for thermally coupled multi-zone buildings. Online, aggregate flexibility is computed through parallel, single-period building-level linear programs followed by the fleet-level closed-form Minkowski summation of power intervals, thereby avoiding operations on high-dimensional full-horizon sets. The resulting endpoint profiles enable a closed-form, time-causal policy that disaggregates any aggregate HVAC power command within the reported interval while preserving recursive feasibility. Case studies show that, relative to the considered fixed full-horizon baselines, the proposed framework captures greater aggregate flexibility by recomputing per-period intervals using the latest state and exogenous information, while maintaining feasible disaggregation under sequentially revealed uncertainty. Its building-separable computation further enables scalable implementation for large fleets of multi-zone buildings.
\end{abstract}

\begin{IEEEkeywords}
Building HVAC systems, reachable set, demand-side aggregation, time-causal operation.
\end{IEEEkeywords}

\section{Introduction}\label{sec:intro}
\IEEEPARstart{B}{uilding} heating, ventilation, and air-conditioning (HVAC) systems are promising demand-side flexible resources for supporting real-time power system balancing \cite{wang2023control}. Owing to building thermal inertia and occupants' comfort tolerance, HVAC power consumption can be flexibly increased or decreased without immediately violating indoor temperature constraints \cite{tian2022real}. This capability is increasingly useful for real-time balancing as renewable generation becomes more variable \cite{zhao2025analytical}. However, individual HVAC systems are typically too small to participate independently in power system operations, while their large population makes direct coordination computationally challenging \cite{liu2025leveraging}. To overcome these barriers, policies such as FERC Order No. 2222 \cite{FERC_Order2222} have enabled \emph{aggregators} to serve as an interface between the power system and distributed energy resources \cite{chen2024wholesale}. In real time, an aggregator coordinates a large HVAC fleet to meet participation thresholds, characterizes the fleet's aggregate flexibility for power system operation, and disaggregates the received aggregate HVAC power command into feasible zone-level control actions \cite{zhao2017geometric}.

Mathematically, each building's flexibility can be characterized by projecting its feasible set from the zone-level state--control space onto the building-level total-power space. The fleet-level aggregate flexibility is then given by the Minkowski sum of the resulting projected sets \cite{liu2024continuous}. However, real-time aggregation of HVAC fleets with multi-period and multi-zone coupling faces two fundamental challenges: efficiently characterizing high-dimensional spatiotemporally coupled flexibility and ensuring feasible period-by-period operation as uncertain exogenous information is revealed sequentially.

First, HVAC flexibility is strongly coupled across time and zones. Thermal dynamics induce multi-period coupling, while heat transfer among zones creates cross-zone coupling within each building \cite{cordova2023aggregate}. Cross-zone coupling yields a high-dimensional feasible set in the zone-level state--control space, making its projection onto the building-level total-power space computationally demanding \cite{liu2025coupling}. Even if these projections are obtained, multi-period coupling makes each projected power set high-dimensional, rendering direct computation of their fleet-level Minkowski sum computationally prohibitive \cite{al2024efficient}.

Second, real-time HVAC fleet operation must remain feasible as uncertain exogenous information is revealed sequentially \cite{rousseau2025uncertainty,tian2022real}. At each period, the aggregator observes the current thermal state and exogenous realization and reports an aggregate flexibility set to the power system operator \cite{liu2026scalable}. The operator selects a power command from this set, which the aggregator then disaggregates into zone-level HVAC actions \cite{liu2025leveraging}. These actions determine the next thermal state and thus future flexibility. Consequently, a currently feasible command may leave no feasible action at a later period under some admissible future exogenous-input sequence \cite{lorca2016multistage}. Therefore, the aggregation and disaggregation policy must be both \emph{time-causal}, using only exogenous information revealed up to the current period, and \emph{recursively feasible}, ensuring that every command in the reported set can be disaggregated while preserving a feasible continuation under every admissible future exogenous-input sequence \cite{zhao2025analytical}. Both properties are essential for reliable real-time aggregation and disaggregation, yet ensuring them is challenging because future system states and exogenous-input realizations remain unknown when current decisions are made.

These challenges call for an aggregation framework that can efficiently handle high-dimensional spatiotemporal coupling while supporting time-causal and recursively feasible real-time operation under sequential information revelation.

\subsection{Related Work}

Existing methods for HVAC flexibility aggregation can be broadly classified into three categories.

\emph{Direct summation} methods \cite{wu2026energy,song2018thermal,han2024analytical} represent individual HVAC systems using single-state thermal-battery models and sum their power and thermal-state limits to obtain aggregate flexibility. This direct-summation construction makes aggregation computationally efficient, while the resulting battery states can be updated sequentially to support time-causal operation. However, these methods generally rely on homogeneous, low-dimensional model structures and therefore cannot capture the heterogeneous, high-dimensional cross-zone coupling of practical multi-zone HVAC systems. Moreover, direct summation generally yields an outer approximation of the true aggregate flexibility \cite{barot2017concise}, which may overestimate the available flexibility and include power trajectories that cannot be feasibly disaggregated. These methods also do not guarantee recursive feasibility as uncertain exogenous information is revealed sequentially.

\emph{Aggregate boundary optimization} methods characterize aggregate flexibility by solving population-level boundary optimization problems to construct inner approximations, such as boxes \cite{chen2021leveraging,chen2024unlock}, ellipsoids \cite{wang2025non,hreinsson2021new}, and vertex-based models \cite{ye2025aggregation,hao2025non}. By explicitly incorporating each building's multi-zone operational model, these methods can accommodate heterogeneity across buildings and cross-zone coupling within buildings. The inner-approximation construction also ensures that the reported aggregate power trajectories admit feasible disaggregation under the modeled forecast conditions. However, jointly handling all HVAC models in a large population-level optimization problem limits scalability \cite{liu2025coupling}. Moreover, although some representations, such as boxes, allow aggregate power to be selected period by period in a time-causal manner, this property alone does not guarantee recursive feasibility under sequentially revealed exogenous uncertainty, thereby raising operational reliability concerns.

\emph{Geometric transformation} methods \cite{zhao2017geometric,al2024efficient,liu2025leveraging,liu2025coupling} first construct an inner approximation of each building's full-horizon flexibility set and then aggregate these approximations using closed-form Minkowski sums. Because the building-level approximations can be constructed independently and in parallel and their Minkowski sums are available in closed form, these methods scale well to large fleets. These methods can also support time-causal power selection at each period through structure-preserving geometric transformations \cite{liu2025leveraging}. 
However, because these sets are typically constructed offline for a fixed full-horizon forecast and are not recomputed from the realized current state and exogenous input, their reported flexibility cannot adapt to newly revealed information. Consequently, the resulting disaggregation does not generally remain recursively feasible under sequential deviations from the forecast \cite{lorca2016multistage}. In addition, constructing a tractable approximation of the full-horizon flexibility set can become increasingly conservative as the scheduling horizon and building state dimension grow \cite{liu2025coupling}.

We also note that recent work has investigated real-time aggregation and disaggregation for storage-like distributed energy resources, such as battery energy storage systems \cite{zheng2024capacity,yu2025flexibility} and electric-vehicle (EV) fleets \cite{yan2023real,xu2025frequency}. These methods sequentially update aggregate flexibility or disaggregation decisions based on observed resource states and EV arrivals, thereby supporting time-causal operation. However, they rely primarily on low-dimensional energy-state dynamics or queueing structures tailored to batteries and electric vehicles. Consequently, they neither directly extend to multi-zone HVAC fleets with high-dimensional cross-zone thermal coupling nor guarantee recursive feasibility for such HVAC fleets as uncertain exogenous information is revealed sequentially.

As summarized in Table~\ref{tab:research_gaps}, existing methods do not simultaneously meet all the key requirements for real-time multi-zone HVAC fleet aggregation. Specifically, a practical framework must accommodate high-dimensional cross-zone coupling, efficiently account for multi-period coupling in real time, operate time-causally as thermal states are observed and exogenous inputs are revealed sequentially, guarantee recursive feasibility of disaggregation for every admissible future exogenous-input sequence, and scale to large fleets.

\begin{table*}[!t]
	\centering
	\caption{Comparison of representative aggregation methods against key requirements for real-time multi-zone HVAC aggregation.}
	\label{tab:research_gaps}
	\setlength{\tabcolsep}{3pt}
	\renewcommand{\arraystretch}{1.10}
	\scriptsize
	\renewcommand{\tabularxcolumn}[1]{m{#1}}
	\begin{tabularx}{\textwidth}{@{}
			>{\raggedright\arraybackslash}m{0.15\textwidth}
			>{\centering\arraybackslash}m{0.1\textwidth}
			>{\centering\arraybackslash}m{0.1\textwidth}
			>{\raggedright\arraybackslash}X
			>{\centering\arraybackslash}m{0.15\textwidth}
			>{\centering\arraybackslash}m{0.15\textwidth}@{}}
		\toprule
		\textbf{Method and Refs.}
		& \makecell{\textbf{Time-causal}\\\textbf{operation}}
		& \makecell{\textbf{Cross-zone}\\\textbf{coupling}}
		& \makecell[l]{\textbf{Multi-period aggregation treatment}\\\textbf{in real-time}}
		& \makecell{\textbf{Recursive feasibility}\\\textbf{under sequential}\\\textbf{information revelation}}
		& \makecell{\textbf{Scalability to}\\\textbf{large multi-zone}\\\textbf{HVAC fleets}} \\
		\midrule
		
		Direct summation \cite{wu2026energy,song2018thermal,han2024analytical}
		& Yes
		& Not applicable
		& Outer battery-state summation
		& Not guaranteed
		& Not applicable \\
		
		Aggregate boundary optimization \cite{chen2021leveraging,chen2024unlock,wang2025non,
			hreinsson2021new,ye2025aggregation,hao2025non}
		& Method-dependent
		& Applicable
		& Fixed fleet-level full-horizon inner approximation
		& Not guaranteed
		& Limited \\
		
		Geometric transformation \cite{zhao2017geometric,al2024efficient,liu2025leveraging,liu2025coupling}
		& Yes
		& Applicable
		& Fixed building-level full-horizon inner approximation plus closed-form Minkowski sum
		& Not guaranteed
		& Yes \\
		
		Real-time battery and EV aggregation \cite{zheng2024capacity,yu2025flexibility,
			yan2023real,xu2025frequency}
		& Yes
		& Not applicable
		& Sequential low-dimensional state update
		& Not guaranteed
		& Not applicable \\
		
		\midrule
		\textbf{This work}
		& \textbf{Yes}
		& \textbf{Yes}
		& \textbf{Per-period reachable-set decomposition plus dynamic exact Minkowski sum of power intervals in real time}
		& \textbf{Guaranteed}
		& \textbf{Yes} \\
		\bottomrule
	\end{tabularx}
\end{table*}

To bridge this gap, we draw on \emph{backward reachability analysis} from control theory \cite{kurzhanskiy2011reach,attar2023data,wetzlinger2025backward}, which characterizes the states from which a dynamic system can satisfy future safety constraints under uncertain exogenous inputs. We adapt this concept to real-time multi-zone HVAC aggregation by using backward reachable sets to encode each building's remaining-horizon feasibility as per-period state constraints. The resulting temporal decomposition enables aggregate flexibility to be computed from current observations and disaggregated period by period while preserving a feasible continuation for every admissible future exogenous-input sequence. To the best of our knowledge, this is the first real-time aggregation framework for multi-zone HVAC fleets that uses backward reachable sets to support time-causal operation and guarantee recursive feasibility under sequential information revelation.

\subsection{Main Contributions}

To summarize, this paper makes the following contributions:

\subsubsection{Offline Decomposition}
We develop an offline temporal decomposition method for multi-zone HVAC systems by adapting backward reachability analysis. The resulting exact backward reachable sets characterize the thermal states from which feasible operation can be maintained over the remaining horizon for every admissible future exogenous-input sequence, thereby providing an equivalent representation of multi-period coupling in terms of per-period state constraints (Proposition~\ref{prop:CausalFeasibility}). To handle high-dimensional cross-zone coupling, we further construct a tailored polytopic inner approximation of each reachable set in affine-image form. This approximation preserves the remaining-horizon feasibility guarantee while retaining a linear representation. We further reformulate its construction as a linear program, enabling tractable offline computation and direct integration into real-time operation (Proposition~\ref{prop:SetContain}).

\subsubsection{Real-Time Policy}
Based on the offline-computed reachable-set representation, we develop a time-causal and scalable real-time aggregation and disaggregation policy for large HVAC fleets. At each period, the policy recomputes aggregate flexibility conditional on the current thermal state and revealed exogenous input through parallel building-level linear programs and the exact closed-form Minkowski sum of the resulting power intervals. This avoids online projection and high-dimensional Minkowski sums of full-horizon power sets. Moreover, any aggregate HVAC power command within the reported interval admits a closed-form zone-level disaggregation that satisfies current-period constraints and keeps each building's next thermal state within its certified reachable set. This preserves recursive feasibility under every admissible future exogenous-input sequence and guarantees the realizability of the reported aggregate flexibility (Proposition~\ref{prop:Disaggregation}).

Although this paper focuses on multi-zone HVAC fleets, the proposed framework can in principle be extended to the real-time aggregation of other distributed energy resources whose dynamics admit linear models and whose operational constraints admit convex polyhedral representations or approximations, such as home batteries and electric vehicles.

\subsection{Paper Organization and Notations}

The remainder of the paper is organized as follows. Section~\ref{sec:ProblemDescription} formulates building- and fleet-level flexibility and the real-time aggregation problem. Section~\ref{sec:AggregationFramework} develops the reachable-set decomposition, tractable inner approximation, and offline--online aggregation and disaggregation policy. Section~\ref{sec:case} presents the numerical results. Section~\ref{sec:conclusion} concludes.

Regarding notation, $\bigoplus$ denotes the Minkowski sum. Throughout, $:= $ means ``is defined as.'' For $\mathcal X\subseteq\mathbb R^{n_x+n_u}$, $\boldsymbol x\in\mathbb R^{n_x}$, and $\boldsymbol u\in\mathbb R^{n_u}$, the projection onto the coordinates associated with $\boldsymbol x$ is
\[
\operatorname{Proj}_{\boldsymbol x}(\mathcal X)
:=\left\{\boldsymbol x\,\middle|\,\exists\,\boldsymbol u\ \text{s.t.}\
[\boldsymbol x^\top,\boldsymbol u^\top]^\top\in\mathcal X\right\}.
\]
For a compact polytope $\mathcal Y$, $\operatorname{Vert}(\mathcal Y)$ denotes its vertex set. The symbol $\boldsymbol I_n$ denotes the $n\times n$ identity matrix, while $\boldsymbol 0_n$ and $\boldsymbol 0_{m\times n}$ denote the $n$-dimensional zero vector and the $m\times n$ zero matrix, respectively. The main notation is summarized in Table~\ref{tab:notation}, with additional ones defined as needed.

\begin{table}[!t]
	\caption{Nomenclature}
	\label{tab:notation}
	\footnotesize
	\centering
	\renewcommand{\arraystretch}{1.05}
	\begin{tabularx}{\columnwidth}{@{}>{\raggedright\arraybackslash}p{0.2\columnwidth}>{\raggedright\arraybackslash}X@{}}
		\toprule
		\textbf{Symbol} & \textbf{Description} \\
		\midrule
		\multicolumn{2}{@{}l}{\textit{Indices and Sets}} \\
		$b\in\mathcal B$ & Building index $b$ and building set $\mathcal B$ \\
		$i\in\mathcal I_b$ & Zone index $i$ and zone set $\mathcal I_b$ of building $b$ \\
		$t\in\mathcal T$ & Period index $t$ and period set $\mathcal T$ \\
		$\mathcal W_{b,t}$ & Admissible set of exogenous-input realizations at period $t$ \\
		$\mathcal Y_b^{\mathrm{ac}}$ & Full-horizon joint state--control feasible set \\
		$\mathcal U_b^{\mathrm{ac}}$ & Full-horizon zone-power flexibility set \\
		$\mathcal U_b^{\mathrm{bld}}$ & Full-horizon building total-power flexibility set \\
		$\mathcal U^{\mathrm{agg}}$ & Full-horizon fleet aggregate-power flexibility set \\
		$\mathcal X_{b,t}^{\mathrm{rch}}$ & Exact backward reachable temperature-state set \\
		$\mathcal X_{b,t}^{\mathrm{pre}}$ & Robust one-step predecessor of the next-period affine set \\
		$\mathcal X_{b,t}^{\mathrm{aff}}$ & Certified affine-image inner approximation of the backward reachable set \\
		$\widetilde{\mathcal Y}_{b,t}^{\mathrm{ac}}$ & Real-time per-period joint state--control feasible set \\
		$\widetilde{\mathcal U}_{b,t}^{\mathrm{bld}}$ & Real-time building total-power flexibility interval \\
		$\widetilde{\mathcal U}_{t}^{\mathrm{agg}}$ & Real-time fleet aggregate-power flexibility interval \\
		\midrule
		\multicolumn{2}{@{}l}{\textit{Parameters}} \\
		$\Delta$ & Duration of each scheduling period \\
		$c_{b,i}^{\mathrm{in}}$ & Thermal capacitance of zone $i$ \\
		$r_{b,ij}^{\mathrm{in}},\ r_{b,i}^{\mathrm{out}}$ & Interzone and zone-to-ambient thermal resistances \\
		$q_{b,t},\ \eta_{b,i}^{\mathrm{gain}}$ & Heat gain and zone heat-gain mapping coefficient \\
		$\eta_{b,i}^{\mathrm{ac}}$ & HVAC efficiency coefficient \\
		$x_{b,i,t}^{\mathrm{set}},\ \beta_{b,i,t}^{\mathrm{set}}$ & Temperature setpoint and comfort tolerance half-band \\
		$\underline u_{b,i}^{\mathrm{ac}},\ \overline u_{b,i}^{\mathrm{ac}}$ & Minimum and maximum HVAC power limits \\
		$\boldsymbol w_{b,t},\ \boldsymbol w_b$ & Period-$t$ exogenous input and full-horizon trajectory \\
		\midrule
		\multicolumn{2}{@{}l}{\textit{Variables}} \\
		$\boldsymbol x_{b,t}^{\mathrm{in}},\ \boldsymbol x_b^{\mathrm{in}}$ & Indoor-temperature state vector and full-horizon trajectory \\
		$\boldsymbol u_{b,t}^{\mathrm{ac}},\ \boldsymbol u_b^{\mathrm{ac}}$ & Zone-level HVAC-power vector and full-horizon trajectory \\
		$u_{b,t}^{\mathrm{bld}},\ \boldsymbol u_b^{\mathrm{bld}}$ & Building total HVAC power and full-horizon trajectory \\
		$u_t^{\mathrm{agg}},\ u_t^{\mathrm{reg}}$ & Aggregate power and aggregate HVAC power command \\
		$u_{b,t}^{\mathrm{bld},\downarrow},\ u_{b,t}^{\mathrm{bld},\uparrow}$ & Lower and upper endpoints of the building total-power interval \\
		$\boldsymbol\gamma_{b,t}^{\mathrm{aff}},\ \boldsymbol\Gamma_{b,t}^{\mathrm{aff}}$ & Translation vector and generator matrix of the affine-image set \\
		$\lambda_t$ & Common endpoint-interpolation coefficient for disaggregation \\
		\bottomrule
	\end{tabularx}
	\vspace{-4mm}
\end{table}

\section{Problem Description}\label{sec:ProblemDescription}

This section formulates the operational model of multi-zone building HVAC fleets and describes the real-time coordination process among building HVAC systems, the HVAC aggregator, and the power system operator.

\subsection{Operational Model of Multi-Zone Building HVAC Fleets}\label{subsec:Problem_Operation}

Fig.~\ref{fig:MultizoneBuilding} shows a schematic of a building equipped with a multi-zone HVAC system. Each building consists of multiple thermal zones whose indoor temperatures are jointly affected by outdoor temperature, heat gain, heat exchange with adjacent zones, HVAC input power, and previous indoor temperatures. As a result, the thermal dynamics are coupled across both zones and time \cite{cui2024decision}.

\begin{figure}[!t]
	\centering
	\includegraphics[width=0.8\linewidth]{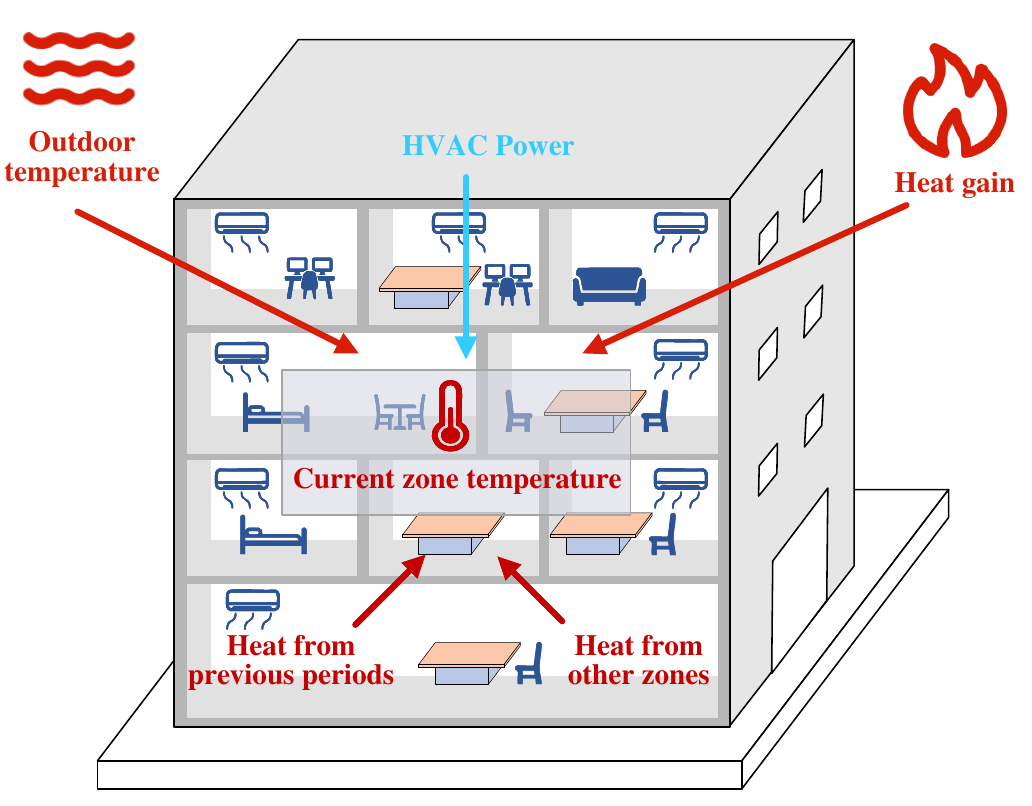}
	\caption{Illustration of a multi-zone building HVAC system.}
	\label{fig:MultizoneBuilding}
\end{figure}

We consider a fleet of $B$ buildings indexed by $b\in \mathcal{B}:=\left\{1,\ldots,B\right\}$, with each building containing $I_b$ thermal zones indexed by $i,j \in \mathcal{I}_b:=\left\{1,\ldots,I_b\right\}$. Time is discretized into $T$ equal-length periods of duration $\Delta$, indexed by $t\in \mathcal{T}:=\left\{1,\ldots,T\right\}$. 
The indoor temperature dynamics are modeled by the finite-difference equation \cite{yang2020HVAC,he2025aggregative,Aslam2024Optimal,wu2026energy}
\begin{equation}
	\eqfitu{
		\begin{aligned}
			c_{b,i}^{\mathrm{in}}\,\frac{x_{b,i,t}^{\mathrm{in}}-x_{b,i,t-1}^{\mathrm{in}}}{\Delta}
			&=\sum_{j\in\mathcal I_b\setminus\{i\}}\frac{x_{b,j,t}^{\mathrm{in}}-x_{b,i,t}^{\mathrm{in}}}{r_{b,ij}^{\mathrm{in}}}
			+\frac{x_{b,t}^{\mathrm{out}}-x_{b,i,t}^{\mathrm{in}}}{r_{b,i}^{\mathrm{out}}}\\
			&\quad+\eta _{b,i}^{\mathrm{gain}}\,q_{b,t}
			-\eta _{b,i}^{\mathrm{ac}}\,u_{b,i,t}^{\mathrm{ac}}.
		\end{aligned}
	}
	\label{eq:HVAC_thermal}
\end{equation}
Here, the state variable $x_{b,i,t}^{\mathrm{in}}$ and control variable $u_{b,i,t}^{\mathrm{ac}}$ denote the indoor temperature and HVAC power of zone $i$ in building $b$ at period $t$, respectively. The parameter $c_{b,i}^{\mathrm{in}}$ is the thermal capacitance, $r_{b,ij}^{\mathrm{in}}$ is the thermal resistance between zones, and $r_{b,i}^{\mathrm{out}}$ is the thermal resistance between zone $i$ and the ambient. The coefficient $\eta_{b,i}^{\mathrm{gain}}$ maps the heat gain to zone $i$, while $\eta_{b,i}^{\mathrm{ac}}$ is the HVAC efficiency coefficient. We initialize the indoor temperature at the first-period setpoint, i.e., $x_{b,i,0}^{\mathrm{in}}=x_{b,i,1}^{\mathrm{set}}, \forall b\in\mathcal B, i\in\mathcal I_b$.

The outdoor temperature $x_{b,t}^{\mathrm{out}}$ and the heat gain $q_{b,t}$ are time-varying exogenous inputs. Define the column vector
$
\boldsymbol w_{b,t}:=
\left[x_{b,t}^{\mathrm{out}},q_{b,t}\right]^\top
\in\mathbb R^2 .
$
To model sequential information revelation in real-time operation, we assume that the realized $\boldsymbol w_{b,t}$ lies in the compact set
\begin{equation}
	\mathcal W_{b,t}:=\left\{\boldsymbol w_{b,t}\,\middle|\,
	\underline{\boldsymbol w}_{b,t}\le\boldsymbol w_{b,t}\le
	\overline{\boldsymbol w}_{b,t}\right\}\subseteq\mathbb R^2,
	\label{eq:HVAC_uncertainty}
\end{equation}
where the vector inequalities are understood componentwise and $\mathcal W_{b,t}$ characterizes the range of exogenous realizations that may be revealed at period $t$. The interval bounds are known offline, while the realization at the current period is revealed before the current control is selected.
\begin{remark}[Scope of exogenous uncertainty]
	Following \cite{yang2020HVAC,he2025aggregative,Aslam2024Optimal,wu2026energy}, we use the lumped term $\eta_{b,i}^{\mathrm{gain}}q_{b,t}$ to represent the net heat gains from occupants, equipment, and solar radiation. Here, $q_{b,t}$ is a building-level heat-gain signal mapped to each zone through $\eta_{b,i}^{\mathrm{gain}}$. Under the adopted thermal model, these sources contribute additively to the zone heat balance, so a net heat-gain input captures their combined effect on temperature dynamics and cooling requirements. This paper focuses on scalable, time-causal aggregation and recursively feasible disaggregation given prescribed uncertainty sets. Detailed source-specific identification, forecasting, and calibration are beyond the scope of this work.
	
	As an extension, additional uncertainty sources can be included by enlarging the admissible range of $q_{b,t}$ when their net effects fit the lumped representation, or by augmenting $\boldsymbol{w}_{b,t}$ and $\mathcal{W}_{b,t}$ while retaining affine dynamics and compact polytopic uncertainty sets. The reachable-set inner approximations are then recomputed for the updated model and uncertainty sets, preserving the recursive-feasibility guarantee for all admissible exogenous-input realizations. The offline computational cost may increase with the number of uncertainty vertices. However, nonlinear or state-dependent disturbance couplings and discrete operating modes require further extensions. For example, window ventilation can introduce multiplicative coupling between uncertain airflow and the indoor--outdoor temperature difference~\cite{ryzhov2019model}. Such extensions are left for future work.
	\label{remark:scope_uncertain}
\end{remark}

For compactness, define the single-period column vectors $\boldsymbol x_{b,t}^{\mathrm{in}}:=\left[x_{b,i,t}^{\mathrm{in}}\right]_{i\in\mathcal I_b}\in\mathbb R^{I_b}$ and $\boldsymbol u_{b,t}^{\mathrm{ac}}:=\left[u_{b,i,t}^{\mathrm{ac}}\right]_{i\in\mathcal I_b}\in\mathbb R^{I_b}$. Then \eqref{eq:HVAC_thermal} can be written compactly as
\begin{equation}
	\eqfitu{
		\boldsymbol A_{b,t}^{\mathrm{eq}}
		\left[
		(\boldsymbol x_{b,t-1}^{\mathrm{in}})^\top,
		(\boldsymbol x_{b,t}^{\mathrm{in}})^\top,
		(\boldsymbol u_{b,t}^{\mathrm{ac}})^\top,
		\boldsymbol w_{b,t}^\top
		\right]^\top
		=\boldsymbol b_{b,t}^{\mathrm{eq}},
	}
	\label{eq:HVAC_thermal_compact}
\end{equation}
where $\boldsymbol A_{b,t}^{\mathrm{eq}}\in\mathbb R^{I_b\times(3I_b+2)}$ and $\boldsymbol b_{b,t}^{\mathrm{eq}}\in\mathbb R^{I_b}$ are the coefficient matrix and vector derived from \eqref{eq:HVAC_thermal}.

We impose the following temperature and power constraints:
\begin{equation}
	\eqfitu{
		x_{b,i,t}^{\mathrm{set}}-\beta _{b,i,t}^{\mathrm{set}}\le
		x_{b,i,t}^{\mathrm{in}}\le
		x_{b,i,t}^{\mathrm{set}}+\beta _{b,i,t}^{\mathrm{set}},
	}
	\label{eq:HVAC_temp_limit}
\end{equation}
\begin{equation}
	\eqfitu{
		\underline{u}_{b,i}^{\mathrm{ac}}
		\le u_{b,i,t}^{\mathrm{ac}}
		\le \overline{u}_{b,i}^{\mathrm{ac}
	}.}
	\label{eq:HVAC_power_limit}
\end{equation}
Here, $x_{b,i,t}^{\mathrm{set}}$ is the indoor temperature setpoint, $\beta_{b,i,t}^{\mathrm{set}}$ is the comfort tolerance half-band, and $\underline{u}_{b,i}^{\mathrm{ac}}$ and $\overline{u}_{b,i}^{\mathrm{ac}}$ are the minimum and maximum HVAC power limits, respectively. Equations \eqref{eq:HVAC_temp_limit} and \eqref{eq:HVAC_power_limit} impose the temperature and HVAC power limits, respectively. Combining these constraints gives
\begin{equation}
	\eqfitu{
		\boldsymbol A_{b,t}^{\mathrm{ieq}}
		\left[
		(\boldsymbol x_{b,t}^{\mathrm{in}})^\top,
		(\boldsymbol u_{b,t}^{\mathrm{ac}})^\top
		\right]^\top
		\le\boldsymbol b_{b,t}^{\mathrm{ieq}},
	}
	\label{eq:HVAC_limit_compact}
\end{equation}
where $\boldsymbol A_{b,t}^{\mathrm{ieq}}\in\mathbb R^{4I_b\times2I_b}$ and $\boldsymbol b_{b,t}^{\mathrm{ieq}}\in\mathbb R^{4I_b}$ are the coefficient matrix and vector derived from \eqref{eq:HVAC_temp_limit} and \eqref{eq:HVAC_power_limit}.

For vectors without a time subscript, define $\boldsymbol x_b^{\mathrm{in}}:=\left[\boldsymbol x_{b,t}^{\mathrm{in}}\right]_{t\in\mathcal T}\in\mathbb R^{I_bT}$, $\boldsymbol u_b^{\mathrm{ac}}:=\left[\boldsymbol u_{b,t}^{\mathrm{ac}}\right]_{t\in\mathcal T}\in\mathbb R^{I_bT}$, and $\boldsymbol w_b:=\left[\boldsymbol w_{b,t}\right]_{t\in\mathcal T}\in\mathbb R^{2T}$. For a fixed full-horizon exogenous trajectory $\boldsymbol w_b$ satisfying $\boldsymbol w_{b,t}\in\mathcal W_{b,t}$, $\forall t\in\mathcal T$, define the building-level augmented feasible set in the joint state--control space as
\begin{equation}
	\eqfitu{
		\mathcal Y_b^{\mathrm{ac}}(\boldsymbol w_b)
		:=\left\{
		\left[(\boldsymbol x_b^{\mathrm{in}})^\top,
		(\boldsymbol u_b^{\mathrm{ac}})^\top\right]^\top
		\,\middle|\,
		\eqref{eq:HVAC_thermal_compact},\ \eqref{eq:HVAC_limit_compact},\ \forall t\in\mathcal T
		\right\}
		\subseteq\mathbb R^{2I_bT}.
	}
	\label{eq:HVAC_flexibility_Y}
\end{equation}

The corresponding zone-power flexibility set is the projection of $\mathcal Y_b^{\mathrm{ac}}(\boldsymbol w_b)$ onto the HVAC power coordinates:
\begin{equation}
	\mathcal U_b^{\mathrm{ac}}(\boldsymbol w_b)
	:=\operatorname{Proj}_{\boldsymbol u_b^{\mathrm{ac}}}
	\!\left(\mathcal Y_b^{\mathrm{ac}}(\boldsymbol w_b)\right)
	\subseteq\mathbb R^{I_bT}.
	\label{eq:HVAC_flexibility_U}
\end{equation}

Because buildings may have different values of $I_b$, the zone-power sets $\mathcal U_b^{\mathrm{ac}}(\boldsymbol w_b)\subseteq\mathbb R^{I_bT}$ have different dimensions and cannot be added across buildings directly. Define instead the total HVAC power of building $b$ by $u_{b,t}^{\mathrm{bld}}:=\sum_{i\in\mathcal I_b}u_{b,i,t}^{\mathrm{ac}}$, its trajectory $\boldsymbol{u}_{b}^{\mathrm{bld}}:=[u_{b,t}^{\mathrm{bld}}]_{t\in \mathcal{T}}\in \mathbb{R} ^T$, and the building total-power trajectory set
\begin{equation}
	\eqfitu{
		\begin{aligned}
			\mathcal U_b^{\mathrm{bld}}(\boldsymbol w_b)
			:=\Bigl\{\boldsymbol u_b^{\mathrm{bld}}\Bigm|\ &
			\exists\,\boldsymbol u_b^{\mathrm{ac}}\in
			\mathcal U_b^{\mathrm{ac}}(\boldsymbol w_b),\\[-0.2em]
			&u_{b,t}^{\mathrm{bld}}=\sum_{i\in\mathcal I_b}u_{b,i,t}^{\mathrm{ac}},
			\ \forall t\in\mathcal T\Bigr\} \subseteq \mathbb{R} ^T.
		\end{aligned}
	}
	\label{eq:Building_flexibility_full_horizon}
\end{equation}

Let $\boldsymbol{w}:=\left[ \boldsymbol{w}_b \right] _{b\in \mathcal{B}}\in \mathbb{R} ^{2BT}$. The fleet aggregate power flexibility set is the dimensionally consistent Minkowski sum
\begin{equation}
	\eqfitu{
		\begin{aligned}
			\mathcal{U} ^{\mathrm{agg}}(\boldsymbol{w})&:=\bigoplus_{b\in \mathcal{B}}{\mathcal{U} _{b}^{\mathrm{bld}}(\boldsymbol{w}_b)}\\
			&=\left\{ \boldsymbol{u}^{\mathrm{agg}}\, \middle| \,\begin{array}{c}
				\exists\,\{\boldsymbol{u}_{b}^{\mathrm{bld}}\in\mathbb R^T\}_{b\in\mathcal B}:\\[-0.2em]
				\underset{\text{fleet-level aggregation summation}}{\underbrace{\boldsymbol{u}^{\mathrm{agg}}=\sum_{b\in \mathcal{B}}{\boldsymbol{u}_{b}^{\mathrm{bld}}}}},\\
				\underset{\text{building-level spatiotemporal coupling}}{\underbrace{\boldsymbol{u}_{b}^{\mathrm{bld}}\in \mathcal{U} _{b}^{\mathrm{bld}}(\boldsymbol{w}_b),\ \forall b\in\mathcal B}}\\
			\end{array} \right\} \subseteq \mathbb{R} ^T.\\
		\end{aligned}
	}
	\label{eq:Agg_flexibility}
\end{equation}

\begin{remark}[Challenge of direct aggregation]
	Directly characterizing $\mathcal U^{\mathrm{agg}}(\boldsymbol w)$ is challenging for two reasons. First, high-dimensional spatiotemporal coupling makes both the building-level characterization in \eqref{eq:Building_flexibility_full_horizon} and the fleet-level construction in \eqref{eq:Agg_flexibility} computationally demanding. Second, in real-time operation, future system states and exogenous realizations are not yet revealed, so aggregate flexibility must be quantified using only the current thermal state and exogenous realization while preserving a feasible continuation over the remaining horizon for every admissible future exogenous-input sequence.
	\label{remark1:challenge}
\end{remark}

\subsection{Real-Time Coordination Process}\label{subsec:Problem_Coordination}

Fig.~\ref{fig:RealtimeProcess} illustrates the real-time, period-by-period coordination among building HVAC systems, the HVAC aggregator, and the power system operator. At the beginning of each period $t\in\mathcal T$, building $b\in\mathcal B$ observes its current thermal state $\boldsymbol x^{\mathrm{in}}_{b,t-1}$ and the current exogenous realization $\boldsymbol w_{b,t}$. Based on this information, the aggregator reports an admissible aggregate power flexibility set $\widetilde{\mathcal U}_t^{\mathrm{agg}}
(\boldsymbol w_t,\boldsymbol x_{t-1}^{\mathrm{in}})$ (defined in \eqref{eq:Agg_flexibility_decoupled}) to the power system operator for period $t$. The reported set should be recursively feasible: every aggregate power $u_{t}^{\mathrm{agg}}\in \widetilde{\mathcal{U} }_{t}^{\mathrm{agg}}(\boldsymbol{w}_t,\boldsymbol{x}_{t-1}^{\mathrm{in}})$ admits a causal feasible continuation over the remaining horizon under all allowed future exogenous realizations. The power system operator then selects an aggregate HVAC power command $u_t^{\mathrm{reg}}$ from the reported set. 
The aggregator disaggregates $u_t^{\mathrm{reg}}$ into zone-level HVAC power commands $u_{b,i,t}^{\mathrm{ac}}$ that satisfy the HVAC dynamics in \eqref{eq:HVAC_thermal}, the temperature and power constraints in \eqref{eq:HVAC_temp_limit} and \eqref{eq:HVAC_power_limit}, respectively, and the aggregate power-balance requirement
$\sum_{b\in\mathcal B}\sum_{i\in\mathcal I_b}u_{b,i,t}^{\mathrm{ac}}
=u_t^{\mathrm{reg}}$, while preserving a causal feasible continuation under all allowed future exogenous realizations.
The same observation, selection, and disaggregation steps are repeated period by period, and the resulting state is carried into the next period.

\begin{figure}[!t]
	\centering
	\includegraphics[width=0.9\linewidth]{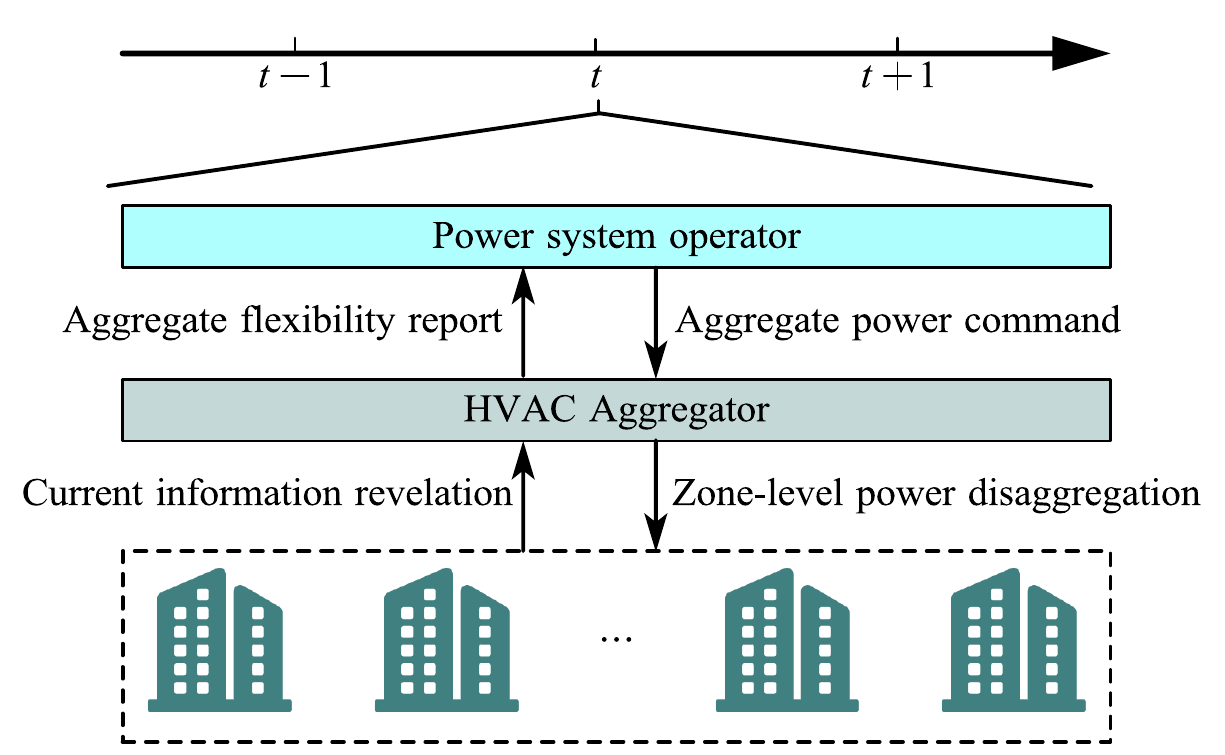}
	\caption{Illustration of the real-time coordination process.}
	\label{fig:RealtimeProcess}
\end{figure}

We therefore aim to develop a time-causal, recursively feasible, and scalable framework for real-time aggregation of multi-zone HVAC fleets under sequential information revelation.

\section{Aggregation Framework}\label{sec:AggregationFramework}
This section develops the proposed offline--online aggregation framework based on reachable-set decomposition. A high-level framework overview is given first, followed by a detailed introduction to each component.

\subsection{Framework Overview}\label{subsec:Framework_Overview}
The framework consists of four connected components:
\begin{enumerate}
	\item A temporal decomposition uses backward reachable sets to replace full-horizon thermal coupling with adjacent-period safe-state conditions, so remaining-horizon feasibility can be equivalently enforced through one-period transitions (Section~\ref{subsec:Framework_Decomposition}).
	\item A tailored affine-image inner approximation converts the multi-zone-coupled reachable-set computation into tractable linear programs, yielding certified safe sets that can be computed efficiently and stored offline (Section~\ref{subsec:Framework_Approximation}).
	\item Using these offline-computed affine sets, a time-causal online policy computes building power intervals in parallel, aggregates them by closed-form Minkowski summation, and disaggregates any admitted aggregate HVAC power command by profile interpolation while preserving recursive feasibility (Section~\ref{subsec:Framework_Aggregation}).
	\item An offline--online algorithm consolidates the preceding components by specifying set construction and certification before the scheduling horizon and period-by-period aggregation and disaggregation during real-time operation (Section~\ref{subsec:Framework_Application}).
\end{enumerate}

\subsection{Reachable-Set Temporal Decomposition}
\label{subsec:Framework_Decomposition}

Although the HVAC model \eqref{eq:HVAC_flexibility_Y} is coupled over the scheduling horizon, its thermal dynamics are Markovian. Physically, the current indoor-temperature vector $\boldsymbol x^{\mathrm{in}}_{b,t-1}$ summarizes the effects of all previous exogenous realizations and HVAC actions on the building's stored thermal energy and therefore quantifies the thermal headroom available for future cooling. Once this state and the current exogenous input $\boldsymbol w_{b,t}$ are known, the earlier history has no additional effect on future feasibility. Remaining-horizon feasibility can therefore be passed backward through adjacent periods as a safe temperature set, rather than reconsidered as a full-horizon trajectory at every real-time decision. This motivates the following backward reachable sets, which encode remaining-horizon feasibility as per-period state constraints and thereby convert multi-period coupling into an equivalent time-causal per-period description.

\begin{definition}[Backward reachable sets \cite{attar2023data}]
	\label{def1:ReachableSet}
	For each $b\in\mathcal B$, define the backward reachable sets
	$\left\{\mathcal X_{b,t}^{\mathrm{rch}}\subseteq\mathbb R^{I_b}\right\}_{t=1,\ldots,T}$.
	Initialize the terminal set using the period-$T$ comfort bounds:
	\begin{equation}
		\adjustbox{max width=\linewidth}{\(\displaystyle
			\mathcal X_{b,T}^{\mathrm{rch}}
			:=\left\{\boldsymbol x_{b,T}^{\mathrm{in}}\,\middle|\,
			x_{b,i,T}^{\mathrm{set}}-\beta_{b,i,T}^{\mathrm{set}}
			\le x_{b,i,T}^{\mathrm{in}}
			\le x_{b,i,T}^{\mathrm{set}}+\beta_{b,i,T}^{\mathrm{set}},\
			\forall i\in\mathcal I_b\right\}.
			\)}
		\label{eq:ReachableSet_Terminal}
	\end{equation}
	For $t=T,T-1,\ldots,2$, define the exact backward recursion by
	\begin{equation}
		\eqfitu{
			\begin{aligned}
				\mathcal X_{b,t-1}^{\mathrm{rch}}
				:=\Bigl\{\boldsymbol x_{b,t-1}^{\mathrm{in}}\ \Bigm|\ 
				&
				x_{b,i,t-1}^{\mathrm{set}}-\beta_{b,i,t-1}^{\mathrm{set}}
				\le x_{b,i,t-1}^{\mathrm{in}},
				\ \forall i\in\mathcal I_b,\\
				&
				x_{b,i,t-1}^{\mathrm{in}}
				\le x_{b,i,t-1}^{\mathrm{set}}+\beta_{b,i,t-1}^{\mathrm{set}},
				\ \forall i\in\mathcal I_b,\\
				&
				\forall\boldsymbol w_{b,t}\in\mathcal W_{b,t},\
				\exists\,\boldsymbol x_{b,t}^{\mathrm{in}},
				\boldsymbol u_{b,t}^{\mathrm{ac}}:\\
				&
				\underset{\mathrm{current}~\mathrm{feasibility}}
				{\underbrace{\eqref{eq:HVAC_thermal_compact},
						\eqref{eq:HVAC_limit_compact}}},
				\
				\underset{\mathrm{recursive}~\mathrm{feasibility}}
				{\underbrace{\boldsymbol x_{b,t}^{\mathrm{in}}
						\in\mathcal X_{b,t}^{\mathrm{rch}}}}
				\Bigr\}.
			\end{aligned}
		}
		\label{eq:ReachableSet_Def}
	\end{equation}
\end{definition}

The terminal box in \eqref{eq:ReachableSet_Terminal} coincides with the period-$T$ comfort constraint in \eqref{eq:HVAC_temp_limit} and therefore introduces no restriction beyond the original HVAC model. At an earlier period, \eqref{eq:ReachableSet_Def} retains exactly those states that satisfy the period-$(t-1)$ comfort bounds and from which, for every possible realization of $\boldsymbol w_{b,t}$, a feasible period-$t$ action can be selected after $\boldsymbol w_{b,t}$ is revealed such that the successor state enters the next reachable set. The explicit period-$(t-1)$ comfort bounds ensure that each reachable set itself contains only physically admissible states. Along a feasible realized trajectory, these bounds are already enforced by the preceding-period constraints. Repeating this one-period handoff preserves feasibility through the terminal period under every admissible future exogenous-input sequence. Formally, we have the following result.

\begin{proposition}[Time-causal feasibility equivalence]
	\label{prop:CausalFeasibility}
	For every $b\in\mathcal B$ and $t\in\{2,\ldots,T\}$,
	$\boldsymbol x_{b,t-1}^{\mathrm{in}}
	\in\mathcal X_{b,t-1}^{\mathrm{rch}}$
	if and only if it satisfies the period-$(t-1)$ comfort bounds and admits a feasible time-causal policy over periods $t$ through $T$. At each period, this policy selects a feasible state--control pair after the current exogenous input is revealed, using only the current state and the revealed exogenous input, and satisfies \eqref{eq:HVAC_thermal_compact} and \eqref{eq:HVAC_limit_compact} for every admissible remaining exogenous-input sequence
	$\left\{\boldsymbol w_{b,\tau}\in\mathcal W_{b,\tau}\right\}_{\tau=t}^{T}$.
\end{proposition}

\begin{proof}
	The recursion has a direct interpretation: at the beginning of period $t$, a state is safe if it satisfies the period-$(t-1)$ comfort bounds and, after any current $\boldsymbol w_{b,t}$ is revealed, a feasible period-$t$ action can be selected whose successor lies in $\mathcal X_{b,t}^{\mathrm{rch}}$. We formalize this interpretation by backward induction.
	
	At $t=T$, membership in
	$\mathcal X_{b,T-1}^{\mathrm{rch}}$ means that the current state satisfies the period-$(T-1)$ comfort bounds and, for every revealed $\boldsymbol w_{b,T}$, a feasible period-$T$ action can be selected whose successor lies in the terminal comfort box
	$\mathcal X_{b,T}^{\mathrm{rch}}$ defined in
	\eqref{eq:ReachableSet_Terminal}. Since no periods remain afterward, this is exactly the existence of a feasible time-causal policy for the final period.
	
	Now fix $t<T$ and assume that
	$\boldsymbol x_{b,t}^{\mathrm{in}}\in\mathcal X_{b,t}^{\mathrm{rch}}$
	if and only if it satisfies the period-$t$ comfort bounds and admits a feasible time-causal policy over periods $t+1$ through $T$.
	
	\emph{Forward direction.}
	If
	$\boldsymbol x_{b,t-1}^{\mathrm{in}}
	\in\mathcal X_{b,t-1}^{\mathrm{rch}}$,
	\eqref{eq:ReachableSet_Def} ensures that the current state satisfies the period-$(t-1)$ comfort bounds and, for every $\boldsymbol w_{b,t}$, a feasible period-$t$ action can be selected whose successor lies in
	$\mathcal X_{b,t}^{\mathrm{rch}}$. The induction hypothesis supplies a feasible time-causal tail policy from that successor over periods $t+1$ through $T$. Since the current action is selected after $\boldsymbol w_{b,t}$ is revealed, combining it with the tail policy yields a feasible time-causal policy over periods $t$ through $T$.
	
	\emph{Reverse direction.}
	Conversely, suppose that
	$\boldsymbol x_{b,t-1}^{\mathrm{in}}$
	satisfies the period-$(t-1)$ comfort bounds and admits a feasible time-causal policy over periods $t$ through $T$. For every current realization $\boldsymbol w_{b,t}$, its period-$t$ action satisfies the current constraints, and its tail policy is feasible from the resulting successor over periods $t+1$ through $T$. The induction hypothesis therefore implies that the successor lies in
	$\mathcal X_{b,t}^{\mathrm{rch}}$. Hence, all conditions in
	\eqref{eq:ReachableSet_Def} hold, proving
	$\boldsymbol x_{b,t-1}^{\mathrm{in}}
	\in\mathcal X_{b,t-1}^{\mathrm{rch}}$.
\end{proof}

Proposition~\ref{prop:CausalFeasibility} shows that, at each period $t$, the single-period constraints \eqref{eq:HVAC_thermal_compact} and \eqref{eq:HVAC_limit_compact} enforce current-period feasibility, while imposing
$\boldsymbol x_{b,t}^{\mathrm{in}}\in\mathcal X_{b,t}^{\mathrm{rch}}$
ensures a feasible time-causal continuation through period $T$. The current-period decision can therefore use only the current state and revealed exogenous input. The reachable set plays a role analogous to that of a value function in dynamic programming: both summarize future consequences to enable per-period decisions, but the reachable set encodes feasibility-to-go for aggregation rather than optimal value-to-go. Accordingly, the proposed reachable-set decomposition can be viewed as a feasibility-oriented counterpart to Bellman's recursion: they recursively characterize feasibility and optimality, respectively, in sequential decision problems under uncertainty. This distinction also separates the proposed framework from approximate dynamic programming. The latter approximates the value function to pursue near-optimal cost, whereas the proposed inner approximation is constructed to lie within the exact reachable set. The resulting flexibility description may therefore be conservative, but its feasibility guarantee is certified rather than approximate.

\subsection{Tractable Inner Approximation}\label{subsec:Framework_Approximation}
Exact evaluation of \eqref{eq:ReachableSet_Def} is computationally
challenging because its robust quantifier structure requires feasible
next-state and control variables for every exogenous realization,
while the predecessor set is obtained by projecting out these
realization-dependent variables. The dimension of this projection
grows with the number of thermally coupled zones, so explicit
projection methods such as Fourier--Motzkin elimination become
increasingly expensive as the model scales \cite{tan2023non}.
We therefore develop an efficient polytopic inner-approximation
method tailored to multi-zone HVAC systems. The resulting inner
approximation preserves the recursive feasibility required by the
online policy, while its linear representation enables tractable
offline computation via
\eqref{eq:ReachableSet_Problem_Reformulate} and online optimization
via \eqref{eq:Building_Problem_Interval}.

\subsubsection{Affine-Image Inner Approximation}

For building $b$, let the unit $\ell_\infty$ ball serve as the base polytope:
\begin{equation}
	\mathcal X_b^{\mathrm{base}}
	:=
	\left\{\boldsymbol z\,\middle|\,
	\|\boldsymbol z\|_\infty\le1\right\}
	=
	\left\{\boldsymbol z\,\middle|\,
	\boldsymbol H_b^{\mathrm{base}}\boldsymbol z
	\le\boldsymbol h_b^{\mathrm{base}}\right\}\subseteq\mathbb R^{I_b},
	\label{eq:ReachableSet_BaseSet}
\end{equation}
where
$\boldsymbol H_b^{\mathrm{base}}\in\mathbb R^{2I_b\times I_b}$
and $\boldsymbol h_b^{\mathrm{base}}\in\mathbb R^{2I_b}$
give its H-representation. We parameterize the approximate
reachable set at period $t$ as the affine image
\begin{equation}
	\begin{aligned}
		\mathcal X_{b,t}^{\mathrm{aff}}
		&:=
		\left\{
		\boldsymbol\Gamma_{b,t}^{\mathrm{aff}}\boldsymbol z
		+\boldsymbol\gamma_{b,t}^{\mathrm{aff}}
		\,\middle|\,
		\boldsymbol z\in\mathcal X_b^{\mathrm{base}}
		\right\} \\
		&=
		\left\{
		\boldsymbol\Gamma_{b,t}^{\mathrm{aff}}\boldsymbol z
		+\boldsymbol\gamma_{b,t}^{\mathrm{aff}}
		\,\middle|\,
		\boldsymbol H_b^{\mathrm{base}}\boldsymbol z
		\le\boldsymbol h_b^{\mathrm{base}}
		\right\}
		\subseteq\mathbb R^{I_b}.
	\end{aligned}
	\label{eq:ReachableSet_AffineBox}
\end{equation}
Here, $\boldsymbol z$ is a normalized base coordinate, $\boldsymbol\gamma_{b,t}^{\mathrm{aff}}\in\mathbb R^{I_b}$ is the translation vector locating the set in temperature space, and $\boldsymbol\Gamma_{b,t}^{\mathrm{aff}}\in\mathbb R^{I_b\times I_b}$ determines the generating directions and their scales. Its off-diagonal entries allow the affine image to represent non-axis-aligned geometry associated with cross-zone thermal coupling. Note that invertibility of $\boldsymbol{\Gamma}_{b,t}^{\mathrm{aff}}$ is not required. Set membership can be imposed through linear constraints using the auxiliary coordinate $\boldsymbol z$, even when the matrix is singular.

Let
$\boldsymbol x_{b,T}^{\mathrm{set}}
:=[x_{b,i,T}^{\mathrm{set}}]_{i\in\mathcal I_b}$
and
$\boldsymbol\beta_{b,T}^{\mathrm{set}}
:=[\beta_{b,i,T}^{\mathrm{set}}]_{i\in\mathcal I_b}$
denote the center and half-width vectors of the terminal set,
respectively. The backward recursion is initialized by
\begin{equation}
	\boldsymbol\Gamma_{b,T}^{\mathrm{aff}}
	:=\operatorname{Diag}
	(\boldsymbol\beta_{b,T}^{\mathrm{set}}),
	\qquad
	\boldsymbol\gamma_{b,T}^{\mathrm{aff}}
	:=\boldsymbol x_{b,T}^{\mathrm{set}},
	\qquad
	\mathcal X_{b,T}^{\mathrm{aff}}
	=\mathcal X_{b,T}^{\mathrm{rch}}.
	\label{eq:ReachableSet_AffineTerminal}
\end{equation}

For $t=T,T-1,\ldots,2$, given the next-period approximation
$\mathcal X_{b,t}^{\mathrm{aff}}$, define its robust one-step
predecessor as
\begin{equation}
	\adjustbox{max width=\linewidth}{\(\displaystyle
		\begin{aligned}
			\mathcal X_{b,t-1}^{\mathrm{pre}}
			:=
			\Bigl\{\boldsymbol x_{b,t-1}^{\mathrm{in}}\ \Bigm|\
			&
			x_{b,i,t-1}^{\mathrm{set}}
			-\beta_{b,i,t-1}^{\mathrm{set}}
			\le x_{b,i,t-1}^{\mathrm{in}},
			\ \forall i\in\mathcal I_b,\\[-0.2em]
			&
			x_{b,i,t-1}^{\mathrm{in}}
			\le x_{b,i,t-1}^{\mathrm{set}}
			+\beta_{b,i,t-1}^{\mathrm{set}},
			\ \forall i\in\mathcal I_b,\\[-0.2em]
			&
			\forall\boldsymbol w_{b,t}\in\mathcal W_{b,t},\
			\exists\,\boldsymbol x_{b,t}^{\mathrm{in}},
			\boldsymbol u_{b,t}^{\mathrm{ac}}:\\[-0.2em]
			&
			\eqref{eq:HVAC_thermal_compact},\
			\eqref{eq:HVAC_limit_compact},\quad
			\boldsymbol x_{b,t}^{\mathrm{in}}
			\in\mathcal X_{b,t}^{\mathrm{aff}}
			\Bigr\}
			\subseteq\mathbb R^{I_b}.
		\end{aligned}\)}
	\label{eq:ReachableSet_Predecessor}
\end{equation}
The backward step constructs an affine image satisfying
\begin{equation}
	\mathcal X_{b,t-1}^{\mathrm{aff}}
	\subseteq
	\mathcal X_{b,t-1}^{\mathrm{pre}}
	\subseteq
	\mathcal X_{b,t-1}^{\mathrm{rch}}.
	\label{eq:ReachableSet_Chain}
\end{equation}
The first inclusion in \eqref{eq:ReachableSet_Chain} is enforced by
the linear containment certificate in
Proposition~\ref{prop:SetContain}, which is imposed in
\eqref{eq:ReachableSet_Problem_Reformulate}. To establish the second
inclusion, suppose that
$\mathcal X_{b,t}^{\mathrm{aff}}
\subseteq\mathcal X_{b,t}^{\mathrm{rch}}$.
By \eqref{eq:ReachableSet_Predecessor}, every state in
$\mathcal X_{b,t-1}^{\mathrm{pre}}$ satisfies the period-$(t-1)$
comfort bounds and admits, for every admissible exogenous realization,
a feasible transition into $\mathcal X_{b,t}^{\mathrm{aff}}$ and hence
into $\mathcal X_{b,t}^{\mathrm{rch}}$. It therefore belongs to
$\mathcal X_{b,t-1}^{\mathrm{rch}}$ by
\eqref{eq:ReachableSet_Def}. Since
\eqref{eq:ReachableSet_AffineTerminal} provides the base case,
\eqref{eq:ReachableSet_Chain} follows by backward induction.

Physically, this nested construction provides a certified
thermal-headroom envelope at each period: keeping the resulting
temperature state within $\mathcal X_{b,t}^{\mathrm{aff}}$ retains
sufficient HVAC actuation capability to accommodate every admissible
future realization of outdoor temperature and heat gain.

To reduce conservatism, we seek a large
$\mathcal X_{b,t-1}^{\mathrm{aff}}$. Its volume satisfies
\begin{equation}
	\operatorname{vol}
	\bigl(\mathcal X_{b,t-1}^{\mathrm{aff}}\bigr)
	=
	\operatorname{vol}
	\bigl(\mathcal X_b^{\mathrm{base}}\bigr)
	\left|
	\det\bigl(\boldsymbol\Gamma_{b,t-1}^{\mathrm{aff}}\bigr)
	\right|.
	\label{eq:ReachableSet_Volume}
\end{equation}
Because maximizing the determinant magnitude is nonconvex,
we use
$\operatorname{Tr}
(\boldsymbol\Gamma_{b,t-1}^{\mathrm{aff}})$,
the sum of its diagonal entries, as a tractable linear size heuristic
\cite{liu2024continuous}:
\begin{equation}
	\eqfitu{
		\begin{aligned}
			\underset{
				\boldsymbol\Gamma_{b,t-1}^{\mathrm{aff}},
				\boldsymbol\gamma_{b,t-1}^{\mathrm{aff}}
			}{\max}\quad
			&\operatorname{Tr}
			\left(\boldsymbol\Gamma_{b,t-1}^{\mathrm{aff}}\right)\\
			\mathrm{s.t.}\quad
			&\mathcal X_{b,t-1}^{\mathrm{aff}}
			\subseteq
			\mathcal X_{b,t-1}^{\mathrm{pre}}.
		\end{aligned}
	}
	\label{eq:ReachableSet_Problem_Original}
\end{equation}
Problem \eqref{eq:ReachableSet_Problem_Original} is not directly
tractable because
$\mathcal X_{b,t-1}^{\mathrm{pre}}$ is defined implicitly through
robust-feasibility quantifiers and projection in \eqref{eq:ReachableSet_Predecessor}. We next
derive a lifted vertex representation and a tractable linear
certificate for the containment condition.

\subsubsection{Lifted Vertex Construction}

For a fixed realization
$\boldsymbol w_{b,t}\in\mathcal W_{b,t}$, define the
scenario-wise feasible polytope
\begin{equation}
	\begin{aligned}
		\mathcal Y_{b,t-1}^{\mathrm{sc}}(\boldsymbol w_{b,t})
		:=\Bigl\{&
		\bigl[
		(\boldsymbol x_{b,t-1}^{\mathrm{in}})^\top,
		\boldsymbol z^\top,
		(\boldsymbol u_{b,t}^{\mathrm{ac}})^\top
		\bigr]^\top
		\ \Bigm|\\[-0.2em]
		&
		x_{b,i,t-1}^{\mathrm{set}}
		-\beta_{b,i,t-1}^{\mathrm{set}}
		\le x_{b,i,t-1}^{\mathrm{in}},
		\quad\forall i\in\mathcal I_b,\\[-0.2em]
		&
		x_{b,i,t-1}^{\mathrm{in}}
		\le x_{b,i,t-1}^{\mathrm{set}}
		+\beta_{b,i,t-1}^{\mathrm{set}},
		\quad\forall i\in\mathcal I_b,\\[-0.2em]
		&
		\boldsymbol z\in\mathcal X_b^{\mathrm{base}},\
		\eqref{eq:HVAC_thermal_compact},\
		\eqref{eq:HVAC_limit_compact}
		\text{ hold with}\\[-0.2em]
		&
		\boldsymbol x_{b,t}^{\mathrm{in}}
		=
		\boldsymbol\Gamma_{b,t}^{\mathrm{aff}}\boldsymbol z
		+\boldsymbol\gamma_{b,t}^{\mathrm{aff}}
		\in\mathcal X_{b,t}^{\mathrm{aff}}
		\Bigr\}
		\subseteq\mathbb R^{3I_b}.
	\end{aligned}
	\label{eq:ReachableSet_Scenario}
\end{equation}
Here, for a given predecessor temperature state
$\boldsymbol x_{b,t-1}^{\mathrm{in}}$ and exogenous-input realization
$\boldsymbol w_{b,t}$, any pair
$(\boldsymbol z,\boldsymbol u_{b,t}^{\mathrm{ac}})$ satisfying the
constraints defining
$\mathcal Y_{b,t-1}^{\mathrm{sc}}(\boldsymbol w_{b,t})$
represents a feasible response. It specifies a feasible successor
temperature state through
$\boldsymbol x_{b,t}^{\mathrm{in}}
=
\boldsymbol\Gamma_{b,t}^{\mathrm{aff}}\boldsymbol z
+
\boldsymbol\gamma_{b,t}^{\mathrm{aff}}$,
together with the corresponding HVAC control. Projecting
$\mathcal Y_{b,t-1}^{\mathrm{sc}}(\boldsymbol w_{b,t})$
onto $\boldsymbol x_{b,t-1}^{\mathrm{in}}$ eliminates these
realization-dependent response variables. Hence,
\begin{equation}
	\mathcal X_{b,t-1}^{\mathrm{pre}}
	=
	\bigcap_{\boldsymbol w_{b,t}\in\mathcal W_{b,t}}
	\operatorname{Proj}_{\boldsymbol x_{b,t-1}^{\mathrm{in}}}
	\left(
	\mathcal Y_{b,t-1}^{\mathrm{sc}}(\boldsymbol w_{b,t})
	\right).
	\label{eq:ReachableSet_Proj_Scenario}
\end{equation}

Let
$V_{b,t}:=|\operatorname{Vert}(\mathcal W_{b,t})|$,
$\mathcal V_{b,t}:=\{1,\ldots,V_{b,t}\}$, and denote the
vertices of $\mathcal W_{b,t}$ by
$\boldsymbol w_{b,t,v}$, $v\in\mathcal V_{b,t}$.
Since $\mathcal W_{b,t}$ is a compact polytope and the constraints
defining
$\mathcal Y_{b,t-1}^{\mathrm{sc}}(\boldsymbol w_{b,t})$
are affine in both the exogenous-input realization and the continuous
response variables, it suffices to enforce feasibility at the vertices
of $\mathcal W_{b,t}$, with a separate feasible response for each
vertex. Physically, these vertices represent the extreme combinations
of outdoor temperature and heat gain. Requiring a feasible response
at every vertex therefore protects the building against the entire
prescribed uncertainty range. Because the response is selected after
the current realization is revealed and all relevant constraints are
affine, convex interpolation of the vertex-specific responses yields
a feasible response for every intermediate realization in the
uncertainty set \cite{attar2023data,yin2019finite}. Accordingly,
\eqref{eq:ReachableSet_Proj_Scenario} reduces to the finite intersection
\begin{equation}
	\mathcal X_{b,t-1}^{\mathrm{pre}}
	=
	\bigcap_{v\in\mathcal V_{b,t}}
	\operatorname{Proj}_{\boldsymbol x_{b,t-1}^{\mathrm{in}}}
	\left(
	\mathcal Y_{b,t-1}^{\mathrm{sc}}
	(\boldsymbol w_{b,t,v})
	\right).
	\label{eq:ReachableSet_Proj_Vertex}
\end{equation}

To avoid explicit projection, stack the shared predecessor state
and the vertex-specific variables as
\begin{equation}
	\begin{aligned}
		\boldsymbol\xi_{b,t-1}^{\mathrm{vrt}}
		:=\Bigl[&
		(\boldsymbol x_{b,t-1}^{\mathrm{in}})^\top,
		\boldsymbol z_{b,t,1}^\top,
		(\boldsymbol u_{b,t,1}^{\mathrm{ac}})^\top,\ldots,\\[-0.2em]
		&\boldsymbol z_{b,t,V_{b,t}}^\top,
		(\boldsymbol u_{b,t,V_{b,t}}^{\mathrm{ac}})^\top
		\Bigr]^\top
		\in\mathbb R^{d_{b,t}},
	\end{aligned}
	\label{eq:ReachableSet_LiftedVector}
\end{equation}
where $d_{b,t}:=I_b(2V_{b,t}+1)$. The corresponding lifted
polytope is
\begin{equation}
	\begin{aligned}
		\mathcal Y_{b,t-1}^{\mathrm{vrt}}
		:=\Bigl\{
		\boldsymbol\xi_{b,t-1}^{\mathrm{vrt}}
		\ \Bigm|\ &
		\bigl[
		(\boldsymbol x_{b,t-1}^{\mathrm{in}})^\top,
		\boldsymbol z_{b,t,v}^\top,
		(\boldsymbol u_{b,t,v}^{\mathrm{ac}})^\top
		\bigr]^\top\\[-0.2em]
		&\in
		\mathcal Y_{b,t-1}^{\mathrm{sc}}
		(\boldsymbol w_{b,t,v}),
		\quad
		\forall v\in\mathcal V_{b,t}
		\Bigr\}\subseteq\mathbb R^{d_{b,t}}.
	\end{aligned}
	\label{eq:ReachableSet_LiftedSet}
\end{equation}
The lifted vector $\boldsymbol\xi_{b,t-1}^{\mathrm{vrt}}$ stacks
the shared predecessor temperature state and the vertex-specific
response variables
$(\boldsymbol z_{b,t,v},\boldsymbol u_{b,t,v}^{\mathrm{ac}})$
for all $v\in\mathcal V_{b,t}$. Any point satisfying
$\boldsymbol\xi_{b,t-1}^{\mathrm{vrt}}
\in\mathcal Y_{b,t-1}^{\mathrm{vrt}}$
therefore certifies that the shared predecessor temperature state
admits a feasible response at every vertex of the exogenous-input
set. Because all defining constraints are linear, this lifted
polytope admits an H-representation
\begin{equation}
	\mathcal Y_{b,t-1}^{\mathrm{vrt}}
	=
	\left\{
	\boldsymbol\xi_{b,t-1}^{\mathrm{vrt}}
	\in\mathbb R^{d_{b,t}}
	\ \middle|\
	\boldsymbol H_{b,t-1}^{\mathrm{vrt}}
	\boldsymbol\xi_{b,t-1}^{\mathrm{vrt}}
	\le
	\boldsymbol h_{b,t-1}^{\mathrm{vrt}}
	\right\}.
	\label{eq:ReachableSet_High_Vertex}
\end{equation}
Here,
$\boldsymbol H_{b,t-1}^{\mathrm{vrt}}
\in\mathbb R^{m_{b,t-1}^{\mathrm{vrt}}\times d_{b,t}}$
and
$\boldsymbol h_{b,t-1}^{\mathrm{vrt}}
\in\mathbb R^{m_{b,t-1}^{\mathrm{vrt}}}$
collect coefficients that are fixed at the current backward step
after substituting the already computed next-period affine set.
The quantity $m_{b,t-1}^{\mathrm{vrt}}$ is the number of lifted
inequalities, with equalities represented by paired inequalities.
Define the projection matrix
\begin{equation}
	\boldsymbol\Gamma_{b,t-1}^{\mathrm{proj}}
	:=
	\left[
	\boldsymbol I_{I_b},
	\boldsymbol 0_{I_b\times 2I_bV_{b,t}}
	\right],
\end{equation}
which extracts the shared predecessor state from the lifted
vector. The predecessor set is then exactly
\begin{equation}
	\mathcal X_{b,t-1}^{\mathrm{pre}}
	=
	\operatorname{Proj}_{\boldsymbol x_{b,t-1}^{\mathrm{in}}}
	\left(\mathcal Y_{b,t-1}^{\mathrm{vrt}}\right)
	=
	\boldsymbol\Gamma_{b,t-1}^{\mathrm{proj}}
	\mathcal Y_{b,t-1}^{\mathrm{vrt}}.
	\label{eq:ReachableSet_Image}
\end{equation}
This lifted representation forms the basis for the linear
containment certificate below.

\subsubsection{Linear Containment Certificate}
The affine set in \eqref{eq:ReachableSet_AffineBox} is parameterized
by a low-dimensional base coordinate, whereas
\eqref{eq:ReachableSet_Image} represents the predecessor set as the
projection of a lifted polytope. Proposition~\ref{prop:SetContain}
introduces an affine reconstruction that maps the base coordinate
associated with each candidate state in
$\mathcal X_{b,t-1}^{\mathrm{aff}}$ to a feasible point in
$\mathcal Y_{b,t-1}^{\mathrm{vrt}}$ whose projection recovers that
candidate state. The set containment can therefore be certified
without explicitly computing the high-dimensional projection.
\begin{proposition}[Linear containment certificate]\label{prop:SetContain}
	The inclusion $\mathcal X_{b,t-1}^{\mathrm{aff}}\subseteq\mathcal X_{b,t-1}^{\mathrm{pre}}$ holds if there exist $\boldsymbol G_{b,t-1}^{\mathrm{aux}}\in\mathbb R^{d_{b,t}\times I_b}$, $\boldsymbol\beta_{b,t-1}^{\mathrm{aux}}\in\mathbb R^{d_{b,t}}$, and $\boldsymbol\Lambda_{b,t-1}^{\mathrm{aux}}\in\mathbb R_+^{m_{b,t-1}^{\mathrm{vrt}}\times2I_b}$ such that
	\begin{equation}
		\begin{aligned}
			\boldsymbol\Gamma_{b,t-1}^{\mathrm{aff}}
			&=\boldsymbol\Gamma_{b,t-1}^{\mathrm{proj}}
			\boldsymbol G_{b,t-1}^{\mathrm{aux}},\\
			-\boldsymbol\gamma_{b,t-1}^{\mathrm{aff}}
			&=\boldsymbol\Gamma_{b,t-1}^{\mathrm{proj}}
			\boldsymbol\beta_{b,t-1}^{\mathrm{aux}},\\
			\boldsymbol\Lambda_{b,t-1}^{\mathrm{aux}}\boldsymbol H_b^{\mathrm{base}}
			&=\boldsymbol H_{b,t-1}^{\mathrm{vrt}}
			\boldsymbol G_{b,t-1}^{\mathrm{aux}},\\
			\boldsymbol\Lambda_{b,t-1}^{\mathrm{aux}}\boldsymbol h_b^{\mathrm{base}}
			&\le\boldsymbol h_{b,t-1}^{\mathrm{vrt}}
			+\boldsymbol H_{b,t-1}^{\mathrm{vrt}}
			\boldsymbol\beta_{b,t-1}^{\mathrm{aux}}.
		\end{aligned}
		\label{eq:ReachableSet_Inclusion_Cond}
	\end{equation}
\end{proposition}

\begin{proof}
	The conditions in \eqref{eq:ReachableSet_Inclusion_Cond}
	are obtained by specializing
	\cite[Theorem~1]{sadraddini2019linear} to the present sets.
	For a direct verification, suppress the subscripts $b,t-1$
	and let
	$\boldsymbol x\in\mathcal X^{\mathrm{aff}}$ be arbitrary.
	By \eqref{eq:ReachableSet_AffineBox}, there exists
	$\boldsymbol\zeta\in\mathcal X^{\mathrm{base}}$ such that
	\begin{equation}
	\boldsymbol x
	=
	\boldsymbol\Gamma^{\mathrm{aff}}\boldsymbol\zeta
	+
	\boldsymbol\gamma^{\mathrm{aff}}.
	\end{equation}
	Define the corresponding lifted point by
	\begin{equation}
		\boldsymbol y
		:=
		\boldsymbol G^{\mathrm{aux}}\boldsymbol\zeta
		-
		\boldsymbol\beta^{\mathrm{aux}}.
	\end{equation}
	Since
	$\boldsymbol H^{\mathrm{base}}\boldsymbol\zeta
	\leq\boldsymbol h^{\mathrm{base}}$
	and
	$\boldsymbol\Lambda^{\mathrm{aux}}\geq\boldsymbol 0$,
	the last two conditions in
	\eqref{eq:ReachableSet_Inclusion_Cond} give
	\begin{equation}
		\begin{aligned}
			\boldsymbol H^{\mathrm{vrt}}\boldsymbol y
			&=
			\boldsymbol H^{\mathrm{vrt}}
			\boldsymbol G^{\mathrm{aux}}\boldsymbol\zeta
			-
			\boldsymbol H^{\mathrm{vrt}}
			\boldsymbol\beta^{\mathrm{aux}}\\
			&=
			\boldsymbol\Lambda^{\mathrm{aux}}
			\boldsymbol H^{\mathrm{base}}\boldsymbol\zeta
			-
			\boldsymbol H^{\mathrm{vrt}}
			\boldsymbol\beta^{\mathrm{aux}}\\
			&\leq
			\boldsymbol\Lambda^{\mathrm{aux}}
			\boldsymbol h^{\mathrm{base}}
			-
			\boldsymbol H^{\mathrm{vrt}}
			\boldsymbol\beta^{\mathrm{aux}}\\
			&\leq
			\boldsymbol h^{\mathrm{vrt}}.
		\end{aligned}
	\end{equation}
	Thus,
	$\boldsymbol y\in\mathcal Y^{\mathrm{vrt}}$.
	The first two conditions in
	\eqref{eq:ReachableSet_Inclusion_Cond} also give
	\begin{equation}
		\begin{aligned}
			\boldsymbol\Gamma^{\mathrm{proj}}\boldsymbol y
			&=
			\boldsymbol\Gamma^{\mathrm{proj}}
			\boldsymbol G^{\mathrm{aux}}\boldsymbol\zeta
			-
			\boldsymbol\Gamma^{\mathrm{proj}}
			\boldsymbol\beta^{\mathrm{aux}}\\
			&=
			\boldsymbol\Gamma^{\mathrm{aff}}\boldsymbol\zeta
			+
			\boldsymbol\gamma^{\mathrm{aff}}
			=
			\boldsymbol x.
		\end{aligned}
	\end{equation}
	Therefore,
	$\boldsymbol x
	\in
	\boldsymbol\Gamma^{\mathrm{proj}}
	\mathcal Y^{\mathrm{vrt}}
	=
	\mathcal X^{\mathrm{pre}}$
	by \eqref{eq:ReachableSet_Image}.
	Since $\boldsymbol x$ was arbitrary,
	$\mathcal X^{\mathrm{aff}}
	\subseteq
	\mathcal X^{\mathrm{pre}}$.
\end{proof}

The affine reconstruction
$\boldsymbol y
=
\boldsymbol G_{b,t-1}^{\mathrm{aux}}\boldsymbol\zeta
-
\boldsymbol\beta_{b,t-1}^{\mathrm{aux}}
\in\mathbb R^{d_{b,t}}$
maps each $I_b$-dimensional base coordinate
$\boldsymbol\zeta$ to a $d_{b,t}$-dimensional feasible point in
the lifted polytope, where
$d_{b,t}=I_b(2V_{b,t}+1)$.
Its first $I_b$ components recover the shared predecessor
temperature state, while the remaining $2I_bV_{b,t}$ components
comprise the successor-state base coordinates and HVAC controls
for all $V_{b,t}$ vertex scenarios. Specifically,
$\boldsymbol G_{b,t-1}^{\mathrm{aux}}$ defines the linear part of
this lifting, with one row for each lifted component and one column
for each base-coordinate component, while
$-\boldsymbol\beta_{b,t-1}^{\mathrm{aux}}$ provides the translation
in the lifted space. Operationally, this construction associates
every candidate predecessor temperature state in
$\mathcal X_{b,t-1}^{\mathrm{aff}}$ with a feasible successor-state
and HVAC-control response for each extreme combination of outdoor
temperature and heat gain. The nonnegative matrix
$\boldsymbol\Lambda_{b,t-1}^{\mathrm{aux}}$ certifies that all lifted
constraints are satisfied for every point in the base polytope.
Hence, robust set containment is certified without explicitly
computing the high-dimensional projection.

Let $\mathcal D_{b,t-1}:=\{\boldsymbol\Gamma_{b,t-1}^{\mathrm{aff}},\boldsymbol\gamma_{b,t-1}^{\mathrm{aff}},\boldsymbol G_{b,t-1}^{\mathrm{aux}},\boldsymbol\beta_{b,t-1}^{\mathrm{aux}},\boldsymbol\Lambda_{b,t-1}^{\mathrm{aux}}\}$. The computable offline problem is the linear program
\begin{equation}
	\eqfitu{
		\begin{aligned}
			\underset{\mathcal D_{b,t-1}}{\max}\quad
			&\operatorname{Tr}(\boldsymbol\Gamma_{b,t-1}^{\mathrm{aff}})\\
			\mathrm{s.t.}\quad
			&\boldsymbol\Gamma_{b,t-1}^{\mathrm{aff}}
			=\boldsymbol\Gamma_{b,t-1}^{\mathrm{proj}}
			\boldsymbol G_{b,t-1}^{\mathrm{aux}},\\
			&-\boldsymbol\gamma_{b,t-1}^{\mathrm{aff}}
			=\boldsymbol\Gamma_{b,t-1}^{\mathrm{proj}}
			\boldsymbol\beta_{b,t-1}^{\mathrm{aux}},\\
			&\boldsymbol\Lambda_{b,t-1}^{\mathrm{aux}}\boldsymbol H_b^{\mathrm{base}}
			=\boldsymbol H_{b,t-1}^{\mathrm{vrt}}
			\boldsymbol G_{b,t-1}^{\mathrm{aux}},\\
			&\boldsymbol\Lambda_{b,t-1}^{\mathrm{aux}}\boldsymbol h_b^{\mathrm{base}}
			\le\boldsymbol h_{b,t-1}^{\mathrm{vrt}}
			+\boldsymbol H_{b,t-1}^{\mathrm{vrt}}
			\boldsymbol\beta_{b,t-1}^{\mathrm{aux}},\\
			&\boldsymbol\Lambda_{b,t-1}^{\mathrm{aux}}\ge\boldsymbol 0.
	\end{aligned}}
	\label{eq:ReachableSet_Problem_Reformulate}
\end{equation}
At each backward step,
$\boldsymbol H_{b,t-1}^{\mathrm{vrt}}$ and
$\boldsymbol h_{b,t-1}^{\mathrm{vrt}}$ are fixed because the
next-period affine set has already been computed. Hence, the objective
and all constraints are linear in $\mathcal D_{b,t-1}$. Starting from
\eqref{eq:ReachableSet_AffineTerminal} and solving
\eqref{eq:ReachableSet_Problem_Reformulate} backward for
$t=T,T-1,\ldots,2$ yields certified affine inner approximations
for all periods $t=1,\ldots,T$.

\subsection{Aggregation and Disaggregation Policy}\label{subsec:Framework_Aggregation}
Using the offline-computed affine inner approximations $\mathcal X_{b,t}^{\mathrm{aff}}$, we then develop a time-causal policy for scalable real-time aggregation and recursively feasible disaggregation.

\subsubsection{Per-Period Reformulation}
The affine-set constraint encodes remaining-horizon feasibility as a constraint on the resulting current-period state. For building $b$, define the per-period joint state--control feasible set as
\begin{equation}
	\begin{aligned}
		\widetilde{\mathcal Y}_{b,t}^{\mathrm{ac}}
		(\boldsymbol w_{b,t},\boldsymbol x_{b,t-1}^{\mathrm{in}})
		:=\Bigl\{&[(\boldsymbol x_{b,t}^{\mathrm{in}})^\top,
		(\boldsymbol u_{b,t}^{\mathrm{ac}})^\top]^\top\ \Bigm|\ 
		\eqref{eq:HVAC_thermal_compact},\ \eqref{eq:HVAC_limit_compact},\\[-0.2em]
		&\boldsymbol x_{b,t}^{\mathrm{in}}\in
		\mathcal X_{b,t}^{\mathrm{aff}}\Bigr\}\subseteq\mathbb R^{2I_b}.
	\end{aligned}
	\label{eq:HVAC_flexibility_Y_decoupled}
\end{equation}
The corresponding building-power flexibility set is
\begin{equation}
	\adjustbox{max width=\linewidth}{\(\displaystyle
		\begin{aligned}
			\widetilde{\mathcal U}_{b,t}^{\mathrm{bld}}
			(\boldsymbol w_{b,t},\boldsymbol x_{b,t-1}^{\mathrm{in}})
			:=\Bigl\{u_{b,t}^{\mathrm{bld}}\ \Bigm|\ &
			\exists\,[ (\boldsymbol x_{b,t}^{\mathrm{in}})^\top,
			(\boldsymbol u_{b,t}^{\mathrm{ac}})^\top]^\top
			\in\widetilde{\mathcal Y}_{b,t}^{\mathrm{ac}}
			(\boldsymbol w_{b,t},\boldsymbol x_{b,t-1}^{\mathrm{in}}),\\[-0.2em]
			&u_{b,t}^{\mathrm{bld}}=
			\sum_{i\in\mathcal I_b}u_{b,i,t}^{\mathrm{ac}}\Bigr\}
			\subseteq\mathbb R.
		\end{aligned}
		\)}
	\label{eq:Building_flexibility_decoupled}
\end{equation}
Using $\boldsymbol w_t:=[\boldsymbol w_{b,t}]_{b\in\mathcal B}$ and
$\boldsymbol x_{t-1}^{\mathrm{in}}
:=[\boldsymbol x_{b,t-1}^{\mathrm{in}}]_{b\in\mathcal B}$,
the fleet-level flexibility set is
\begin{equation}
	\adjustbox{max width=\linewidth}{\(\displaystyle
		\begin{aligned}
			\widetilde{\mathcal U}_t^{\mathrm{agg}}
			(\boldsymbol w_t,\boldsymbol x_{t-1}^{\mathrm{in}})
			&:=\bigoplus_{b\in\mathcal B}
			\widetilde{\mathcal U}_{b,t}^{\mathrm{bld}}
			(\boldsymbol w_{b,t},\boldsymbol x_{b,t-1}^{\mathrm{in}})\\
			&=\Bigl\{u_t^{\mathrm{agg}}\ \Bigm|\
			\exists\,\{u_{b,t}^{\mathrm{bld}}\}_{b\in\mathcal B}:\
			u_t^{\mathrm{agg}}
			=\sum_{b\in\mathcal B}u_{b,t}^{\mathrm{bld}},\\[-0.2em]
			&\hspace{5em}
			u_{b,t}^{\mathrm{bld}}
			\in\widetilde{\mathcal U}_{b,t}^{\mathrm{bld}}
			(\boldsymbol w_{b,t},\boldsymbol x_{b,t-1}^{\mathrm{in}}),
			\ \forall b\in\mathcal B
			\Bigr\}\subseteq\mathbb R.
		\end{aligned}
		\)}
	\label{eq:Agg_flexibility_decoupled}
\end{equation}

\subsubsection{Real-Time Aggregation Policy}
For the feasible current states considered in real-time operation,
$\widetilde{\mathcal Y}_{b,t}^{\mathrm{ac}}
(\boldsymbol w_{b,t},\boldsymbol x_{b,t-1}^{\mathrm{in}})$
is nonempty. Because it is defined by affine equalities and
inequalities, with the state constrained to the bounded and closed
set $\mathcal X_{b,t}^{\mathrm{aff}}$ and the HVAC powers subject to
bounded limits, it is a compact polytope. Its scalar linear image is
therefore a closed interval. Suppressing the arguments for brevity, we have
\begin{equation}
	\widetilde{\mathcal U}_{b,t}^{\mathrm{bld}}
	=\left[u_{b,t}^{\mathrm{bld},\downarrow},
	u_{b,t}^{\mathrm{bld},\uparrow}\right].
	\label{eq:Building_interval}
\end{equation}
The superscripts $\downarrow$ and $\uparrow$ label the lower and upper
endpoints and their corresponding optimal state and zone-power
profiles. Compactness ensures that both endpoints are attained by
solving the following linear programs:
\begin{equation}
	\begin{aligned}
		u_{b,t}^{\mathrm{bld},\downarrow}
		&=\underset{[(\boldsymbol x_{b,t}^{\mathrm{in},\downarrow})^\top,
			(\boldsymbol u_{b,t}^{\mathrm{ac},\downarrow})^\top]^\top
			\in\widetilde{\mathcal Y}_{b,t}^{\mathrm{ac}}(\boldsymbol w_{b,t},\boldsymbol x_{b,t-1}^{\mathrm{in}})}{\min}
		\ \sum_{i\in\mathcal I_b}u_{b,i,t}^{\mathrm{ac},\downarrow},\\
		u_{b,t}^{\mathrm{bld},\uparrow}
		&=\underset{[(\boldsymbol x_{b,t}^{\mathrm{in},\uparrow})^\top,
			(\boldsymbol u_{b,t}^{\mathrm{ac},\uparrow})^\top]^\top
			\in\widetilde{\mathcal Y}_{b,t}^{\mathrm{ac}}(\boldsymbol w_{b,t},\boldsymbol x_{b,t-1}^{\mathrm{in}})}{\max}
		\ \sum_{i\in\mathcal I_b}u_{b,i,t}^{\mathrm{ac},\uparrow}.
	\end{aligned}
	\label{eq:Building_Problem_Interval}
\end{equation}
Once the building-level intervals are computed, the aggregate flexibility set in \eqref{eq:Agg_flexibility_decoupled} exactly reduces to the closed-form Minkowski sum of intervals:
\begin{equation}
	\begin{aligned}
		\widetilde{\mathcal U}_t^{\mathrm{agg}}
		&=\bigoplus_{b\in\mathcal B}
		\left[u_{b,t}^{\mathrm{bld},\downarrow},
		u_{b,t}^{\mathrm{bld},\uparrow}\right]\\
		&=\left[
		\sum_{b\in\mathcal B}u_{b,t}^{\mathrm{bld},\downarrow},
		\sum_{b\in\mathcal B}u_{b,t}^{\mathrm{bld},\uparrow}
		\right].
	\end{aligned}
	\label{eq:Agg_Calculation_Interval}
\end{equation}
The aggregator reports \eqref{eq:Agg_Calculation_Interval} as the time-causal admissible aggregate power range for period $t$. This closed-form interval summation avoids high-dimensional online Minkowski-sum calculations and enables scalable aggregation.

\subsubsection{Real-Time Disaggregation Policy}
After receiving $u_t^{\mathrm{reg}}\in\widetilde{\mathcal U}_t^{\mathrm{agg}}(\boldsymbol w_t,\boldsymbol x_{t-1}^{\mathrm{in}})$, the aggregator must map it to zone-level HVAC powers while preserving aggregate power balance, current feasibility, and recursive feasibility. The interval structure and convexity of the building-level feasible sets yield the following direct construction without any additional fleet-level online optimization problem.

\begin{proposition}[Closed-form feasible disaggregation]\label{prop:Disaggregation}
	Given any
	\begin{equation}
		u_t^{\mathrm{reg}}
		\in
		\left[
		\sum_{b\in\mathcal B}u_{b,t}^{\mathrm{bld},\downarrow},
		\sum_{b\in\mathcal B}u_{b,t}^{\mathrm{bld},\uparrow}
		\right],
	\end{equation}
	define
	\begin{equation}
		\eqfitu{
			\lambda_t=
			\frac{
				u_t^{\mathrm{reg}}-\sum_{b\in\mathcal B}u_{b,t}^{\mathrm{bld},\downarrow}
			}{
				\sum_{b\in\mathcal B}u_{b,t}^{\mathrm{bld},\uparrow}
				-
				\sum_{b\in\mathcal B}u_{b,t}^{\mathrm{bld},\downarrow}
			}
			\in[0,1].
		}
		\label{eq:Disagg_coef_building}
	\end{equation}
	Then, for every building $b\in\mathcal B$, choose the zone-level power and temperature profile as
	\begin{equation}
		\eqfitu{
			\begin{aligned}
				\begin{bmatrix}
					\boldsymbol x_{b,t}^{\mathrm{in}}\\
					\boldsymbol u_{b,t}^{\mathrm{ac}}
				\end{bmatrix}
				&=(1-\lambda_t)
				\begin{bmatrix}
					\boldsymbol x_{b,t}^{\mathrm{in},\downarrow}\\
					\boldsymbol u_{b,t}^{\mathrm{ac},\downarrow}
				\end{bmatrix}
				+\lambda_t
				\begin{bmatrix}
					\boldsymbol x_{b,t}^{\mathrm{in},\uparrow}\\
					\boldsymbol u_{b,t}^{\mathrm{ac},\uparrow}
				\end{bmatrix}.
			\end{aligned}
		}
		\label{eq:Disagg_profile_zone}
	\end{equation}
	The resulting allocation is time-causal, feasible at period $t$, exactly tracks $u_t^{\mathrm{reg}}$, and is recursively feasible over the remaining horizon.
\end{proposition}
\begin{proof}
	For each building $b$, both endpoint joint profiles in \eqref{eq:Building_Problem_Interval} belong to the convex polytope $\widetilde{\mathcal Y}_{b,t}^{\mathrm{ac}}(\boldsymbol w_{b,t},\boldsymbol x_{b,t-1}^{\mathrm{in}})$. Since $\lambda_t\in[0,1]$, the convex combination in \eqref{eq:Disagg_profile_zone} belongs to the same set and therefore satisfies the current thermal dynamics and operating limits. Linearity also gives
	\begin{equation}
		\sum_{i\in\mathcal I_b}u_{b,i,t}^{\mathrm{ac}}
		=(1-\lambda_t)u_{b,t}^{\mathrm{bld},\downarrow}
		+\lambda_t u_{b,t}^{\mathrm{bld},\uparrow}.
		\label{eq:Disagg_proof}
	\end{equation}
	Summing \eqref{eq:Disagg_proof} over $b\in\mathcal B$ and substituting \eqref{eq:Disagg_coef_building} yields
	\begin{equation}
		\sum_{b\in\mathcal B}\sum_{i\in\mathcal I_b}u_{b,i,t}^{\mathrm{ac}}
		=u_t^{\mathrm{reg}},
		\label{eq:Disagg_balance}
	\end{equation}
	so the aggregate HVAC power command is tracked exactly.
	Membership in
    $\widetilde{\mathcal Y}_{b,t}^{\mathrm{ac}}$
    also gives
    $\boldsymbol x_{b,t}^{\mathrm{in}}
    \in\mathcal X_{b,t}^{\mathrm{aff}}$.
    For $t<T$, the backward construction in
    \eqref{eq:ReachableSet_Chain} gives
    $\mathcal X_{b,t}^{\mathrm{aff}}
    \subseteq\mathcal X_{b,t}^{\mathrm{pre}}$.
    Hence, by \eqref{eq:ReachableSet_Predecessor}, for every
    $\boldsymbol w_{b,t+1}\in\mathcal W_{b,t+1}$, the next-period feasible set
    $\widetilde{\mathcal Y}_{b,t+1}^{\mathrm{ac}}
    (\boldsymbol w_{b,t+1},\boldsymbol x_{b,t}^{\mathrm{in}})$
    is nonempty and contains a feasible state--control pair whose successor
    state lies in $\mathcal X_{b,t+1}^{\mathrm{aff}}$.
    At period $t+1$, the same convex-interpolation argument maps any aggregate
    power command within the newly reported interval to feasible building-level
    profiles and keeps each resulting state in
    $\mathcal X_{b,t+1}^{\mathrm{aff}}$.
    Repeating this argument recursively yields feasible disaggregation through
    period $T$ for every admissible future exogenous-input sequence.
    For $t=T$, the terminal-set constraint in
    \eqref{eq:ReachableSet_AffineTerminal} completes the horizon.
    Finally, the endpoint problems, coefficient calculation, and interpolation
    at each period use only the current state, revealed exogenous input, and
    current aggregate HVAC power command. Therefore, the resulting policy is
    time-causal and recursively feasible.
\end{proof}

The key intuition behind Proposition~\ref{prop:Disaggregation} is that the common coefficient $\lambda_t$ acts as a relative interpolation factor, rather than assigning the same absolute power to every building or zone. Specifically, relative to its lower endpoint, building $b$ provides the adjustment $\lambda_t(u_{b,t}^{\mathrm{bld},\uparrow}-u_{b,t}^{\mathrm{bld},\downarrow})$, so a building with a wider feasible interval contributes a larger adjustment. At the zone level, \eqref{eq:Disagg_profile_zone} interpolates between the two building-specific feasible state--control profiles obtained from \eqref{eq:Building_Problem_Interval}, rather than between the nameplate power limits of individual HVAC units. Because these endpoint problems incorporate the thermal dynamics, current temperatures, revealed exogenous conditions, thermal parameters, power limits, and zone-specific HVAC efficiencies $\eta_{b,i}^{\mathrm{ac}}$, local efficiency differences affect the interval widths and endpoint profiles used in the allocation.

\begin{remark}[Scope of the disaggregation policy]
	The closed-form disaggregation policy \eqref{eq:Disagg_profile_zone} is designed to prioritize computational efficiency and guaranteed feasibility rather than optimize a fleet-level economic objective. Although local HVAC efficiency differences are reflected in the building-level feasible sets and endpoint profiles, the common coefficient $\lambda_t$ does not explicitly prioritize chillers according to their operating costs. Cost-aware, comfort-aware, or priority-aware dispatch can instead be obtained by solving a fleet-level optimization problem over the same building-level feasible sets, subject to the aggregate power-balance constraint. Retaining these feasible-set constraints would preserve current-period and recursive feasibility, but the additional coordination problem would increase online computation and communication requirements. This reflects the trade-off between richer allocation objectives and the computational simplicity and real-time scalability of the proposed closed-form policy. 
	\label{remark:DisaggregationScope}
\end{remark}

\subsection{Implementation Procedure}\label{subsec:Framework_Application}
Algorithm~\ref{alg:Implementation} summarizes the implementation procedure of our offline--online aggregation framework. The offline stage is executed once before the scheduling horizon, whereas the real-time stage is repeated at every period, so offline reachable-set construction is separated from period-by-period online operation. This provides an efficient, time-causal, recursively feasible, and scalable workflow.
\begin{algorithm}[!t]
	\caption{Offline--online implementation procedure}
	\label{alg:Implementation}
	\footnotesize
	\begin{algorithmic}[1]
		\Statex \textbf{Offline stage (once before the scheduling horizon)}
		\ForAll{$b\in\mathcal B$ \textbf{ in parallel}}
		\State Initialize $\mathcal X_{b,T}^{\mathrm{aff}}$ using \eqref{eq:ReachableSet_AffineTerminal}.
		\For{$t=T,T-1,\ldots,2$}
		\State Form $\mathcal Y_{b,t-1}^{\mathrm{vrt}}$ using \eqref{eq:ReachableSet_High_Vertex}.
		\State Solve the containment linear program \eqref{eq:ReachableSet_Problem_Reformulate}.
		\State Store the certified $\mathcal X_{b,t-1}^{\mathrm{aff}}$.
		\EndFor
		\EndFor
		\Statex \textbf{Real-time stage (repeat for $t=1,\ldots,T$)}
		\For{$t=1,2,\ldots,T$}
		\ForAll{$b\in\mathcal B$ \textbf{ in parallel}}
		\State Observe $\boldsymbol x_{b,t-1}^{\mathrm{in}}$ and $\boldsymbol w_{b,t}$.
		\State Solve \eqref{eq:Building_Problem_Interval} and store both endpoint profiles.
		\EndFor
		\State Compute and report $\widetilde{\mathcal U}_t^{\mathrm{agg}}$ using \eqref{eq:Agg_Calculation_Interval}.
		\State Receive $u_t^{\mathrm{reg}}\in\widetilde{\mathcal U}_t^{\mathrm{agg}}$.
		\State Compute the interpolation coefficient as
		\Statex \hspace{\algorithmicindent}$\displaystyle
		\lambda_t\gets
		\frac{u_t^{\mathrm{reg}}-\sum_{b\in\mathcal B}u_{b,t}^{\mathrm{bld},\downarrow}}
		{\sum_{b\in\mathcal B}u_{b,t}^{\mathrm{bld},\uparrow}
			-\sum_{b\in\mathcal B}u_{b,t}^{\mathrm{bld},\downarrow}}$.
		\ForAll{$b\in\mathcal B$ \textbf{ in parallel}}
		\State Interpolate the joint profile using \eqref{eq:Disagg_profile_zone}.
		\EndFor
		\State Apply the selected zone powers and update $\boldsymbol x_{b,t}^{\mathrm{in}}$ for all $b$.
		\EndFor
	\end{algorithmic}
\end{algorithm}

\begin{remark}[Building-level separability and scalability]
	The thermal dynamics \eqref{eq:HVAC_thermal} and operating constraints \eqref{eq:HVAC_temp_limit}-\eqref{eq:HVAC_power_limit} are separable across buildings. Consequently, the offline containment linear programs in \eqref{eq:ReachableSet_Problem_Reformulate} and the online endpoint linear programs in \eqref{eq:Building_Problem_Interval} can be solved independently across buildings and in parallel. The aggregator then forms the fleet-level interval in \eqref{eq:Agg_Calculation_Interval} by summing the building-level endpoints. After receiving the aggregate HVAC power command, it broadcasts the common coefficient in \eqref{eq:Disagg_coef_building}, which each building uses to interpolate its stored endpoint profiles. Thus, increasing the fleet size adds only independent building-level linear programs and lightweight scalar coordination. Neither a fleet-scale online linear program nor a high-dimensional Minkowski-sum computation is required. With sufficient parallel computing resources, this separable structure supports efficient real-time aggregation and disaggregation of large HVAC fleets.
	\label{remark:ComputationalScalability}
\end{remark}

\section{Case Studies}\label{sec:case}
This section presents numerical case studies that evaluate the aggregate flexibility, disaggregation reliability, and computational scalability of the proposed framework under high-dimensional spatiotemporal coupling and sequentially revealed exogenous uncertainty.

\subsection{Simulation Setup}\label{subsec:Case_Setup}
All HVAC parameters sampled from Table~\ref{tab:HVACParameter}, including thermal resistances, capacitances, efficiencies, setpoints, comfort bands, and power limits, are sampled independently across buildings and zones before the offline stage.
Here, $\operatorname{Unif}[a,b]$ denotes the uniform distribution on $[a,b]$. Parameter values are randomly sampled from these distributions to reflect heterogeneity. The time-varying comfort half-band is defined as
$
\beta_{b,i,t}^{\mathrm{set}}
=\delta_{b,i,t}\bar\beta_{b,i}^{\mathrm{set}},
$
where $\delta_{b,i,t}$ is a time-varying comfort-tolerance coefficient, and the base comfort half-band $\bar\beta_{b,i}^{\mathrm{set}}$ is sampled from the distribution specified in Table~\ref{tab:HVACParameter}.

The simulations are conducted in MATLAB on a laptop with an Intel Core Ultra~7 265U processor and 32~GB of RAM. Gurobi \cite{gurobi} is used to solve the linear programs, and MATLAB's \texttt{parfor} is configured to use 12 workers for parallel computation.

\begin{table}[!t]
	\centering
	\caption{HVAC parameter ranges used in the simulations.}
	\label{tab:HVACParameter}
	\footnotesize
	\setlength{\tabcolsep}{2pt}
	\renewcommand{\arraystretch}{1.08}
	\begin{tabularx}{\linewidth}{@{}C{0.15\linewidth}YC{0.30\linewidth}@{}}
		\toprule
		\textbf{Param.} & \textbf{Description} & \textbf{Range} \\
		\midrule
		$r_{b,ij}^{\mathrm{in}}$ & Interzone thermal resistance & $\operatorname{Unif}[2,6]~{}^\circ\mathrm C/\mathrm{kW}$ \\
		$r_{b,i}^{\mathrm{out}}$ & Zone-to-ambient thermal resistance & $\operatorname{Unif}[1,3]~{}^\circ\mathrm C/\mathrm{kW}$ \\
		$c_{b,i}^{\mathrm{in}}$ & Zone thermal capacitance & $\operatorname{Unif}[1,5]~\mathrm{kWh}/{}^\circ\mathrm C$ \\
		$\eta_{b,i}^{\mathrm{gain}}$ & Zone heat-gain mapping coefficient& $\operatorname{Unif}[0.2,0.8]~\mathrm{kW}/\mathrm{p.u.}$ \\
		$\eta_{b,i}^{\mathrm{ac}}$ & HVAC efficiency coefficient & $\operatorname{Unif}[2,4]$ \\
		$x_{b,i,t}^{\mathrm{set}}$ & Indoor-temperature setpoint & $\operatorname{Unif}[20,24]~{}^\circ\mathrm C$ \\
		$\bar\beta_{b,i}^{\mathrm{set}}$ & Base comfort half-band & $\operatorname{Unif}[1,3]~{}^\circ\mathrm C$ \\
		$\underline u_{b,i}^{\mathrm{ac}}$ & Minimum HVAC power & $\operatorname{Unif}[0.3,0.7]~\mathrm{kW}$ \\
		$\overline u_{b,i}^{\mathrm{ac}}$ & Maximum HVAC power & $\operatorname{Unif}[2,4]~\mathrm{kW}$ \\
		\bottomrule
	\end{tabularx}
\end{table}

\subsection{An Illustrative Example}\label{subsec:Case_Illustrative}

We first consider a thermally coupled two-zone HVAC system to illustrate the role of reachable sets. The system consists of one building ($B=1$) with $I_1=2$ zones and operates over $T=4$ periods of duration $\Delta=1$~h. The temperatures of zones $\mathrm A$ and $\mathrm B$ are denoted by $x^{\mathrm A}$ and $x^{\mathrm B}$, respectively. Both zones use the four-period comfort-tolerance coefficient sequence
$(\delta_{1,i,t})_{t=1}^{4}=(1,0.5,1,0.5)$ for $i\in\mathcal I_1$. At each period, the outdoor temperature and the building-level equivalent net heat gain are assumed to lie in the uncertainty intervals $[36,40]{}^\circ\mathrm C$ and $[0.4,0.6]$~p.u., respectively. The remaining parameters are sampled from the distributions specified in Table~\ref{tab:HVACParameter}.

\begin{figure}[!t]
	\centering
	\includegraphics[width=\linewidth]{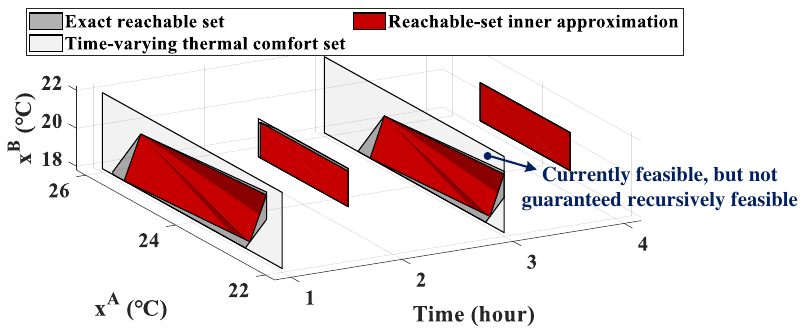}
	\caption{Exact reachable sets and certified inner approximations.}
	\label{fig:IllustrativeExample}
\end{figure}

Fig.~\ref{fig:IllustrativeExample} displays the exact reachable sets and proposed inner approximations for all four periods. At each period, the outlined rectangle represents the direct comfort bounds in \eqref{eq:HVAC_temp_limit} and captures only instantaneous comfort feasibility. Satisfying those bounds alone does not guarantee a feasible continuation because current HVAC decisions affect future temperatures. The gray exact reachable set contains states that satisfy the comfort constraints at the displayed period and admit a feasible time-causal continuation under every admissible future exogenous-input sequence. It is therefore generally a strict subset of the corresponding comfort set. Consequently, a state may satisfy the current comfort constraints yet lie outside the reachable set, in which case a feasible continuation cannot be guaranteed under sequential information revelation. The blue point illustrates such a state.

The red certified inner approximations are contained in the corresponding exact reachable sets, as established by \eqref{eq:ReachableSet_Chain} using Proposition~\ref{prop:SetContain}, while capturing the dominant cross-zone thermal coupling structure. In this two-zone example, the exact reachable sets can be computed through exact projection using Fourier--Motzkin elimination (FME). However, this computation quickly becomes intractable as the number of zones increases, as shown in Table~\ref{tab:ComputationTime}.

These results demonstrate the role of reachable sets in guaranteeing recursive feasibility in time-causal multi-period operation. They also show that the proposed inner approximations provide a tractable representation of cross-zone coupling.

\subsection{Comparison With Other Methods}\label{subsec:Case_Comparison}

We next compare the proposed framework with an oracle benchmark and two representative aggregation baselines. The scheduling horizon spans 7:00--19:00 to represent working hours, during which the buildings are occupied and their HVAC systems are in service, and consists of $T=12$ periods of duration $\Delta=1$~h. Unless otherwise specified, the fleet contains $B=100$ buildings. For each building, the number of zones $I_b$ is sampled uniformly from the discrete set $\{4,\ldots,12\}$, and the thermal parameters are sampled independently from the ranges in Table~\ref{tab:HVACParameter}. For each building, a separate sparse, connected, undirected graph is randomly generated to specify the thermally coupled zone pairs.
The solid lines in Fig.~\ref{fig:OutdoorParameter} show the forecasts of the outdoor temperature $x_{b,t}^{\mathrm{out}}$ and the normalized heat gain $q_{b,t}$, while the shaded bands indicate their admissible uncertainty ranges. 
 The time-varying comfort-tolerance coefficient $\delta_{b,i,t}$ for each zone and period is uniformly sampled within the hourly bounds shown in Fig.~\ref{fig:ToleranceParameter}. For every building and zone, this coefficient sequence is sampled before the offline stage and then held fixed across all four cases and all Monte Carlo scenarios. The black line indicates the nominal coefficient, and the gray band indicates the corresponding sampling bounds.

\begin{figure}[!t]
	\centering
	\includegraphics[width=0.95\linewidth]{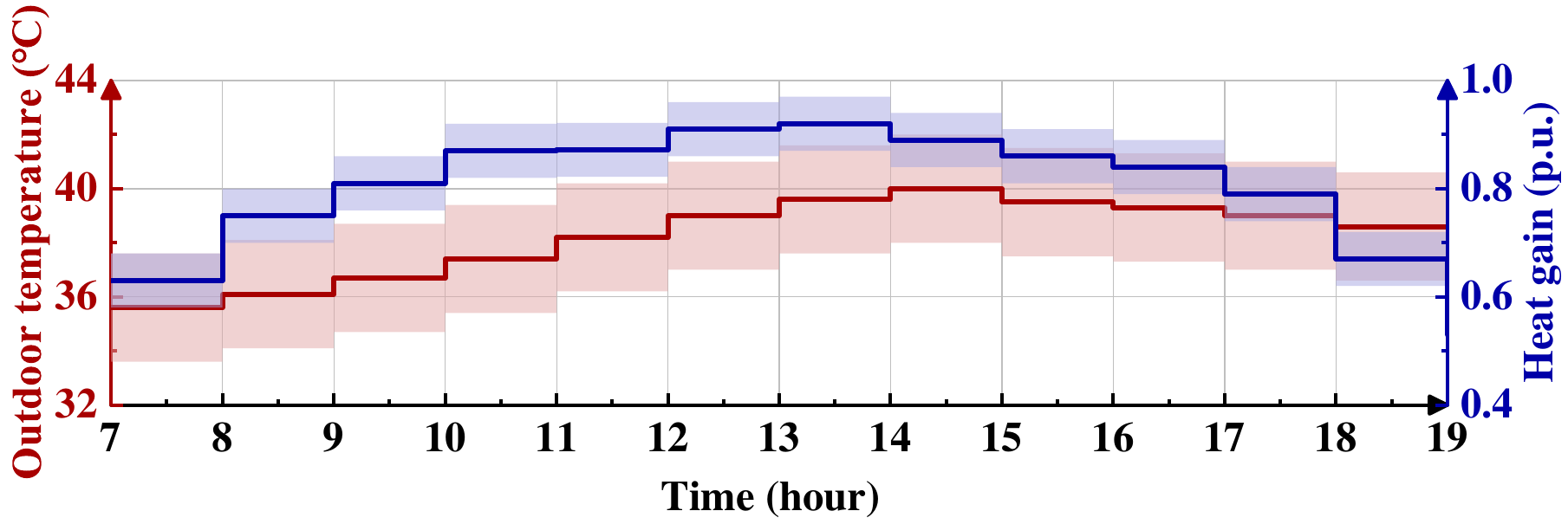}
	\caption{Outdoor-temperature and normalized heat-gain forecasts and admissible uncertainty ranges.}
	\label{fig:OutdoorParameter}
\end{figure}

\begin{figure}[!t]
	\centering
	\includegraphics[width=0.95\linewidth]{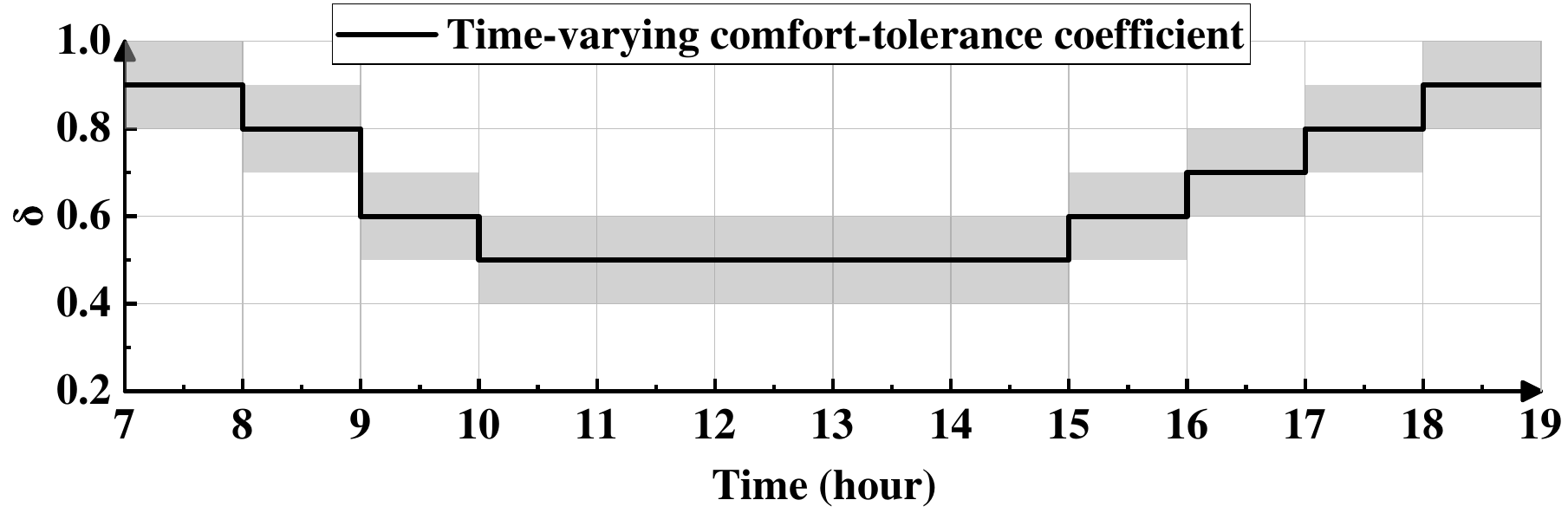}
	\caption{Nominal comfort-tolerance coefficient and hourly sampling bounds.}
	\label{fig:ToleranceParameter}
\end{figure}

Under this setup, we compare the following four cases.

\textbf{Case 0 (Oracle benchmark):} The complete realized exogenous-input trajectory is assumed to be known in hindsight. At each period, the flexibility interval is recomputed from the current state and the known remaining exogenous-input trajectory. This reference relies on future realizations and therefore cannot be implemented online. It serves only as an oracle benchmark for realizable flexibility.

\textbf{Case 1 (Proposed framework):} Certified inner approximations of the backward reachable sets are computed using the offline stage of Algorithm~\ref{alg:Implementation}. At each period, the online stage uses only the current state and revealed exogenous input to compute aggregate flexibility and disaggregate the admitted power command according to the proposed policy.

\textbf{Case 2 (Geometric-transformation baseline):} This baseline represents the class of methods in \cite{zhao2017geometric,al2024efficient,liu2025leveraging,liu2025coupling}. It uses the forecast exogenous-input trajectory shown in Fig.~\ref{fig:OutdoorParameter} to construct fixed full-horizon building-level inner approximations offline. These approximations are combined through a closed-form Minkowski sum to obtain the fleet-level aggregate flexibility set. The resulting set is used period by period during online operation. It is not reconstructed after the actual states and exogenous inputs are revealed.

\textbf{Case 3 (Boundary-optimization baseline):} This baseline represents the methods in \cite{chen2021leveraging,chen2024unlock}. Using the forecast exogenous-input trajectory in Fig.~\ref{fig:OutdoorParameter}, it constructs a fixed box inner approximation of the full-horizon aggregate flexibility set offline and applies its per-period bounds sequentially without online reconstruction.

Cases 1--3 are all operated sequentially without access to future realized exogenous inputs. In Cases 2 and 3, the period-$t$ flexibility bounds are taken from their respective fixed forecast-based full-horizon sets and applied when period $t$ is reached. Their online implementation is therefore time-causal in information use, but their reported flexibility sets are not information-adaptive: these sets are not recomputed from the realized current state and exogenous input. In contrast, Case 1 recomputes each reported interval from the latest state and revealed exogenous input while using backward reachable sets to account for unrevealed future uncertainty.

\subsubsection{Aggregation and Disaggregation Results}

\begin{figure}[!t]
	\centering
	\includegraphics[width=\linewidth]{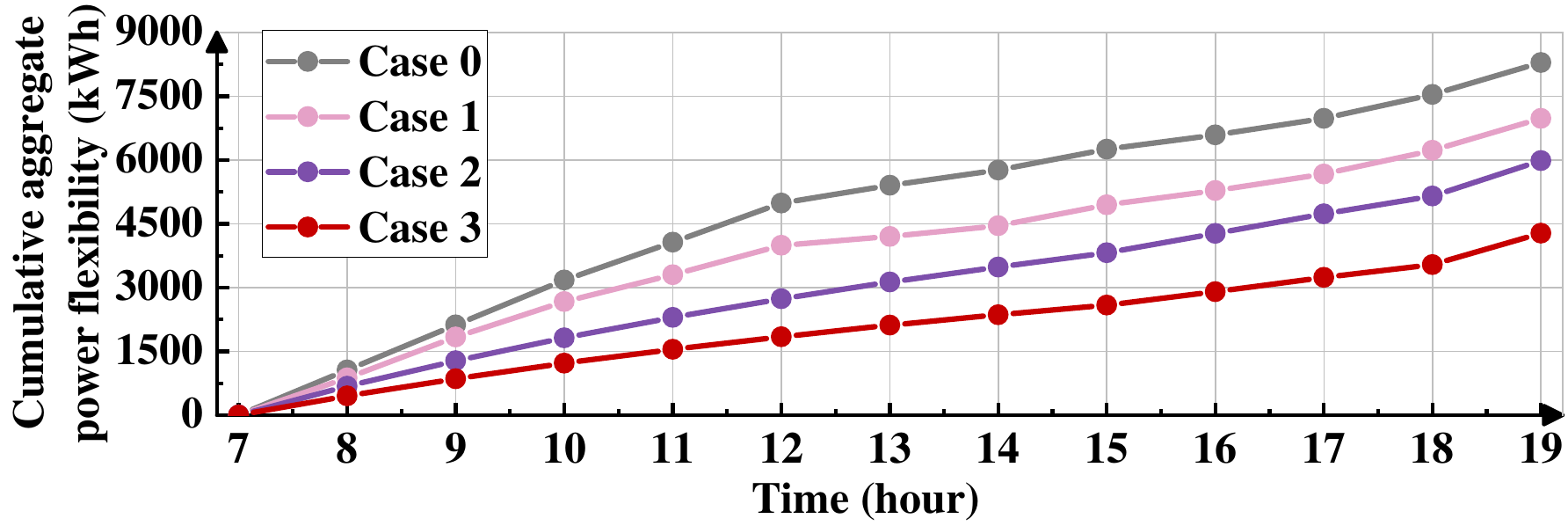}
	\caption{Cumulative aggregate flexibility under the four cases.}
	\label{fig:CumulativeFlexibility}
\end{figure}

\begin{figure}[!t]
	\centering
	\includegraphics[width=\linewidth]{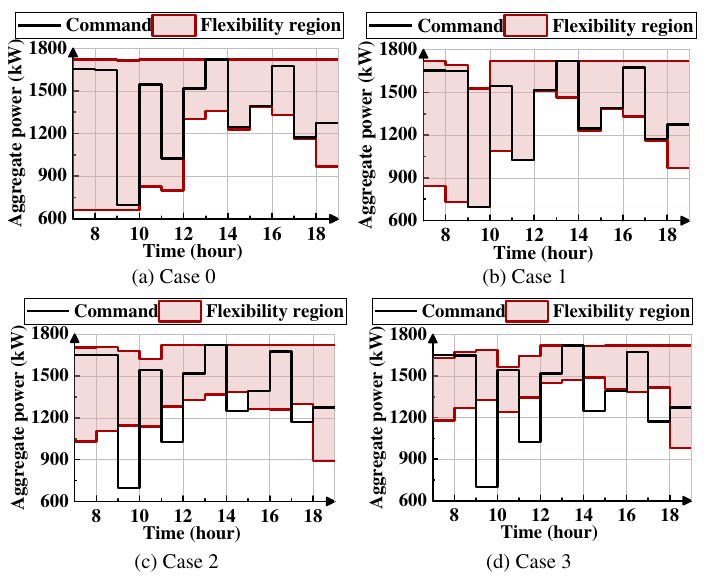}
	\caption{Period-by-period aggregate flexibility intervals and power commands.}
	\label{fig:EachPeriodFlexibility}
\end{figure}

\begin{figure}[!t]
	\centering
	\includegraphics[width=\linewidth]{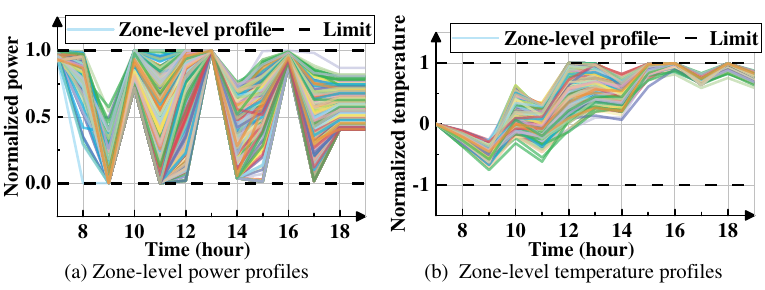}
	\caption{Normalized zone-level disaggregation trajectories in Case~1.}
	\label{fig:DisaggregateResult}
\end{figure}

Fig.~\ref{fig:CumulativeFlexibility} compares the cumulative aggregate flexibility under the four cases. At each period $t$, the cumulative aggregate flexible energy is $\Delta$ times the sum of the reported aggregate-flexibility interval widths through period $t$ and is therefore expressed in kWh. Case~1 provides substantially greater flexibility than Cases~2 and~3 while remaining below the oracle benchmark in Case~0. Case~3 is the most conservative because its fixed box approximation cannot closely capture the geometry of the true multi-period-coupled flexibility set. This mismatch leads to substantial flexibility loss. Case 2 captures temporal coupling more closely through polytope-based geometric transformations. However, its tractable full-horizon inner approximation remains conservative and is not recomputed from the realized current state and exogenous input.

The proposed method reduces this representation conservatism by encoding remaining-horizon feasibility as per-period state constraints using backward reachable sets. At each period, it recomputes a one-dimensional power interval for each building conditional on the latest information, and obtains the fleet-level interval through an exact closed-form Minkowski sum. This temporal decomposition therefore avoids constructing full-horizon power-set approximations, enables information-adaptive aggregation, and eliminates high-dimensional online set operations. The remaining gap between Cases~0 and~1 has two principal sources: the robust flexibility reservation required for unknown future exogenous-input realizations and the conservatism of the certified inner approximations of the exact backward reachable sets. The effects of exogenous uncertainty and increasing zone dimensionality are further examined in Figs.~\ref{fig:Sensitivity} and~\ref{fig:NumZoneImpact}, respectively.

Fig.~\ref{fig:EachPeriodFlexibility} shows the corresponding period-by-period flexibility intervals and power commands. During the restrictive midday periods (hours 12--16), larger outdoor-temperature and heat-gain realizations, together with tighter comfort tolerances, reduce the range of feasible HVAC operation. The resulting aggregate flexibility intervals are therefore narrower. Case~1 also reports narrower intervals than Case~0 during the preceding hours (hours 7--11). This restriction preserves sufficient thermal headroom for feasible operation during the subsequent restrictive periods under sequential information revelation. By contrast, Cases 2 and 3 remain narrow over much of the horizon. Because their fixed full-horizon sets are not recomputed from the realized current state and exogenous input, they cannot adapt the reported flexibility to newly revealed information. This pattern reflects both the conservatism of their full-horizon approximations and the absence of information-adaptive online recomputation, rather than an explicit state-dependent reservation for future feasibility.

Fig.~\ref{fig:DisaggregateResult}  reports the power and temperature trajectories of all zones across the 100-building fleet under the disaggregation policy in \eqref{eq:Disagg_profile_zone}. For each zone, define the normalized power and temperature as
\begin{equation}
	\eqfitu{
		\widehat u_{b,i,t}^{\mathrm{ac}}
		:=
		\bigl(u_{b,i,t}^{\mathrm{ac}}
		-\underline u_{b,i}^{\mathrm{ac}}\bigr)
		\big/
		\bigl(\overline u_{b,i}^{\mathrm{ac}}
		-\underline u_{b,i}^{\mathrm{ac}}\bigr)
		\in[0,1].
	}
	\label{eq:Case_normalized_power}
\end{equation}
\begin{equation}
	\eqfitu{
		\widehat x_{b,i,t}^{\mathrm{in}}
		:=
		\bigl(x_{b,i,t}^{\mathrm{in}}
		-x_{b,i,t}^{\mathrm{set}}\bigr)
		\big/
		\beta_{b,i,t}^{\mathrm{set}}
		\in[-1,1].
	}
	\label{eq:Case_normalized_temperature}
\end{equation}
For each building, the common coefficient $\lambda_t$ in \eqref{eq:Disagg_coef_building} interpolates between its lower and upper feasible endpoint profiles. Thus, buildings with wider feasible power intervals contribute larger adjustments, while their zone-level power and temperature profiles follow the corresponding endpoint interpolation. As shown in Fig.~\ref{fig:DisaggregateResult}, the zone-level trajectories within each panel exhibit similar overall temporal trends, while their exact values differ because of heterogeneous thermal parameters and operating constraints. Nevertheless, all normalized power and temperature trajectories remain within their limits throughout real-time operation, consistent with the recursive-feasibility guarantee in Proposition~\ref{prop:Disaggregation}.

These results show that the proposed framework provides greater aggregate flexibility while maintaining feasible zone-level operation.

\subsubsection{Flexibility and Reliability Under Sequentially Revealed Exogenous Uncertainty}

We generate $2{,}000$ Monte Carlo scenarios, each comprising a real-time exogenous-input trajectory and a target aggregate HVAC power-command trajectory. All four cases are evaluated on the same set of scenarios and therefore face identical realized exogenous-input and target aggregate power-command trajectories.
At each period, each case selects the feasible aggregate power closest to the common target within its reported aggregate-flexibility interval, subject to the realized thermal dynamics and the power and temperature constraints. The selected power is implemented using that case’s disaggregation procedure to advance its state. 
The infeasible scenario ratio is defined as the percentage of scenarios in which, at some period, no zone-level HVAC action satisfying the realized thermal dynamics and operating constraints exists from that case's current state for any aggregate power.
 For each scenario, total flexibility is calculated as the sum of the reported aggregate-flexibility interval widths over all periods. The total flexibility ratio for each of Cases~1--3 is defined as its total flexibility divided by that of Case~0, expressed as a percentage.

\begin{figure}[!t]
	\centering
	\includegraphics[width=\linewidth]{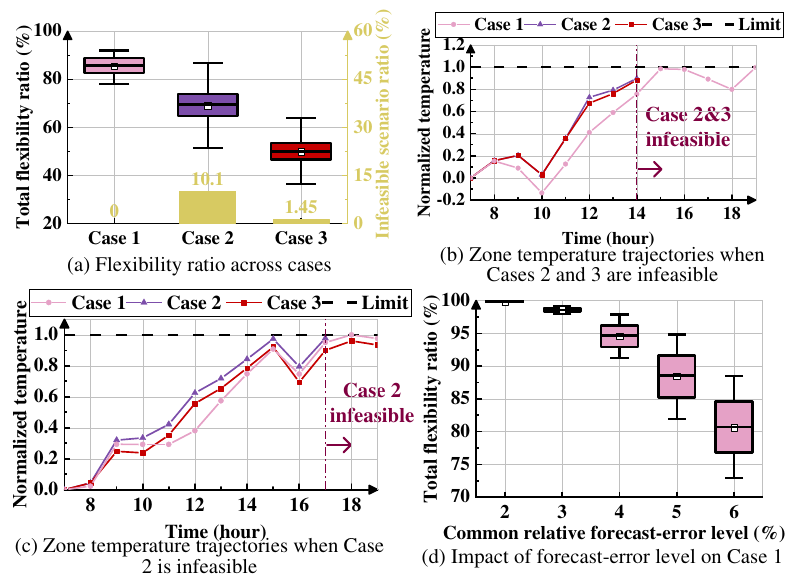}
	\caption{Flexibility and feasibility under sequential information revelation.}
	\label{fig:Sensitivity}
\end{figure}

Fig.~\ref{fig:Sensitivity}(a) compares the distributions of the total flexibility ratio and the corresponding infeasible scenario ratios. Case~1 achieves the highest median flexibility ratio and a tighter distribution than Cases~2 and~3. Its infeasible scenario ratio is $0\%$, compared with $10.10\%$ for Case~2 and $1.45\%$ for Case~3. Because all sampled exogenous inputs lie within the prescribed uncertainty sets, the reachable-set construction guarantees feasible disaggregation for Case~1. By contrast, Cases 2 and 3 rely on fixed forecast-based sets that are neither recomputed from the realized current state and exogenous input nor certified against all admissible future exogenous-input sequences. Sequential deviations from the forecast can therefore invalidate subsequent disaggregation. The lower infeasible ratio of Case~3 results from its smaller and more conservative box, rather than from an explicit treatment of sequential information revelation.

Figs.~\ref{fig:Sensitivity}(b) and~\ref{fig:Sensitivity}(c) show zone-temperature trajectories from two representative infeasible scenarios. In Fig.~\ref{fig:Sensitivity}(b), Cases~2 and~3 are closer to the upper comfort boundary than Case~1 at hour~14. At hour~15, neither case admits any feasible zone-level HVAC action satisfying the realized thermal dynamics and operating constraints. The limited thermal headroom is therefore associated with disaggregation failure in both cases. The resulting state of Case~1 lies in the certified reachable set, so a feasible continuation exists. In Fig.~\ref{fig:Sensitivity}(c), Case~2 encounters such a feasibility failure, whereas Case~3 retains a feasible zone-level solution. At hour~17, the temperature under Case~3 is lower than that under Case~1, leaving greater thermal headroom. This additional margin results from its conservative box rather than an explicit treatment of sequential information revelation, reducing the risk of disaggregation failure at the cost of aggregate flexibility. In contrast, Case~1 utilizes more flexibility while retaining a certified feasible continuation.

Fig.~\ref{fig:Sensitivity}(d) examines the effect of the common relative forecast-error level on Case~1. Here, the forecast-error level denotes the common relative half-width of the uncertainty intervals around both the outdoor-temperature and heat-gain forecasts. As this level increases from $2\%$ to $6\%$, the distribution of the total flexibility ratio shifts downward and becomes more dispersed. Broader uncertainty intervals require larger thermal margins to preserve recursive feasibility over a wider range of future realizations, thereby reducing the reportable flexibility. 
These results characterize the relationship between prescribed forecast-error bounds and reportable flexibility under the adopted thermal model. An economic assessment of forecasting investment additionally requires models linking investment to error bounds and valuing the resulting flexibility. Developing these models is beyond the scope of this study.

Overall, the Monte Carlo results show that the proposed framework preserves feasible disaggregation under sequential information revelation while retaining substantially more flexibility than the baseline methods. It therefore achieves a better flexibility--reliability trade-off.

\subsubsection{Computational Scalability With Fleet Size}

\begin{figure}[!t]
	\centering
	\includegraphics[width=\linewidth]{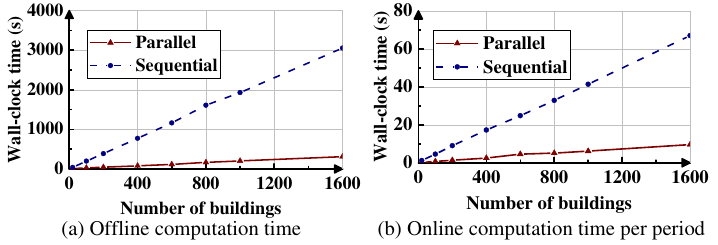}
	\caption{Sequential and parallel wall-clock times versus fleet size.}
	\label{fig:CPUtime}
\end{figure}

Fig.~\ref{fig:CPUtime} compares the sequential and parallel computation times for fleets ranging from 20 to 1,600 buildings. Offline time is the total time required to construct the reachable-set approximations for all buildings over the scheduling horizon, whereas online time is the mean per-period computation time for fleet-level aggregation. Both stages are separable across buildings, as described in Remark~\ref{remark:ComputationalScalability}.
For $B = 1{,}600$, offline computation takes $3{,}051.56$~s sequentially and $311.04$~s with 12 parallel workers. The corresponding online times are $67.00$ and $9.56$~s per period. Under the fixed 12-worker configuration, parallel computation times increase approximately linearly with fleet size. These results demonstrate the computational benefits of the building-separable architecture for large-fleet aggregation. Additional processor cores or distributed workers could further reduce computation time.

\subsubsection{Impact of Zone Count}

Table~\ref{tab:ComputationTime} reports the single-building computation times as the zone count increases. FME denotes exact reachable-set projection using Fourier--Motzkin elimination. A time limit of 24~h is imposed on FME: the $>86{,}400$ entry indicates that the computation did not terminate within this limit, and the larger instances marked by ``--'' were therefore not attempted. 
The offline times for FME and the proposed method refer to the total computation time required to construct the exact reachable sets and their certified inner approximations, respectively, for a single building over the full scheduling horizon. The proposed online time is the mean per-period aggregation time.

\begin{table}[!t]
	\centering
	\caption{Single-building computation times versus zone count.}
	\label{tab:ComputationTime}
	\footnotesize
	\setlength{\tabcolsep}{2.8pt}
	\begin{tabular}{@{}cccc@{}}
		\toprule
		\textbf{Number of zones}
		& \makecell{\textbf{Offline FME}\\\textbf{(s)}}
		& \makecell{\textbf{Offline proposed}\\\textbf{(s)}}
		& \makecell{\textbf{Online proposed}\\\textbf{(s)}} \\
		\midrule
		4  & 27.48       & 1.433  & 0.041 \\
		8  & $>$86,400 & 1.881  & 0.042 \\
		12 & --          & 2.749  & 0.041 \\
		16 & --          & 4.851  & 0.040 \\
		20 & --          & 8.444 & 0.047 \\
		24 & --          & 17.521 & 0.044 \\
		\bottomrule
	\end{tabular}
\end{table}

As the zone count increases from 4 to 24, the proposed offline time increases from $1.433$ to $17.521$~s, while its online time remains below $0.05$~s. In contrast, exact FME exceeds the 24~h limit at eight zones, for which the proposed offline and online times are only $1.881$ and $0.042$~s, respectively. This difference arises because FME may generate an exponentially growing number of intermediate inequalities when projecting a high-dimensional polytope \cite{tan2023non}, whereas the proposed method avoids explicit projection by solving the linear program in \eqref{eq:ReachableSet_Problem_Reformulate} and retaining a compact affine-image representation.

Fig.~\ref{fig:NumZoneImpact} evaluates the corresponding fleet-level aggregation performance. Following the preceding Monte Carlo setup, the fleet contains $B=100$ buildings and is evaluated over the same $2{,}000$ scenarios. For each plotted zone count, all buildings have that number of zones, while their thermal parameters are independently resampled across buildings and zones using the distributions in Table~\ref{tab:HVACParameter}.

\begin{figure}[!t]
	\centering
	\includegraphics[width=\linewidth]{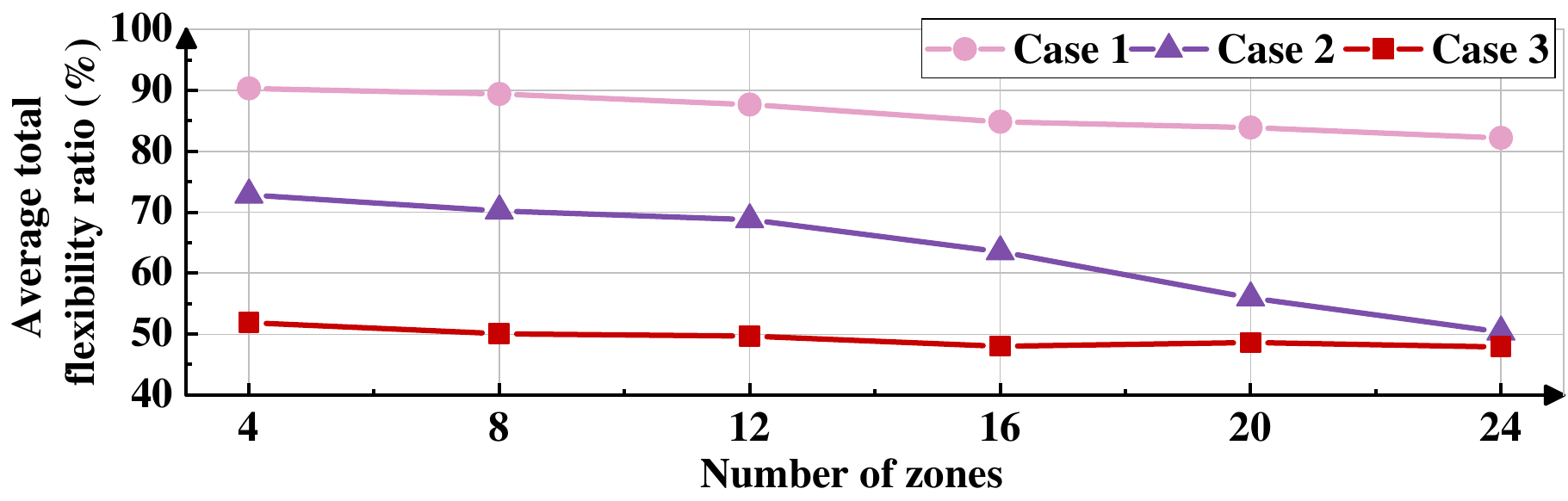}
	\caption{Average total flexibility ratio versus zone count.}
	\label{fig:NumZoneImpact}
\end{figure}

As shown in Fig.~\ref{fig:NumZoneImpact}, the average total flexibility ratio of Case~1 gradually decreases as the reachable-set dimension increases, which is consistent with greater approximation conservatism in higher dimensions. Nevertheless, it remains above $80\%$ at 24 zones and consistently exceeds the ratios of Cases~2 and~3. Together, Table~\ref{tab:ComputationTime} and Fig.~\ref{fig:NumZoneImpact} show that the proposed inner approximation avoids the prohibitive cost of exact projection while maintaining tractable computation and more aggregate flexibility than the full-horizon baselines.

\section{Conclusion}\label{sec:conclusion}
This paper develops an offline--online reachable-set decomposition framework for real-time aggregation of multi-zone HVAC fleets with high-dimensional spatiotemporal coupling under sequentially revealed exogenous uncertainty. Backward reachable sets encode remaining-horizon feasibility as per-period state constraints, while tailored polytopic inner approximations provide tractable representations of cross-zone thermal coupling. Based on the offline-computed reachable sets, the proposed real-time policy performs building-level aggregation in parallel, computes fleet-level power intervals through closed-form Minkowski summation, and disaggregates aggregate power commands through a time-causal policy with a recursive-feasibility guarantee.

Numerical results demonstrate that the proposed framework reports greater aggregate flexibility than the considered fixed full-horizon baselines while maintaining feasible disaggregation under sequentially revealed uncertainty. This improvement reflects both the avoidance of conservative full-horizon power-set approximations and the ability to recompute each per-period interval from the latest state and exogenous input. Meanwhile, the reachable-set construction reserves sufficient thermal headroom to maintain feasibility under every admissible future exogenous-input sequence, yielding a better flexibility–reliability trade-off. The building-separable architecture enables effective parallel computation and shows the potential to support large HVAC fleets. Moreover, the proposed inner approximation avoids the prohibitive projection cost of exact reachable-set computation as the number of zones increases.

This work opens some important directions for future study.
One direction is to decompose the lumped heat-gain signal into source-level components and quantify how uncertainty in these components affects reportable flexibility. Further extensions could address state-dependent or nonlinear disturbance couplings and discrete operating modes.
Another direction is to incorporate economic, comfort, and priority objectives into the disaggregation policy while preserving recursive feasibility, and to investigate the resulting trade-off between richer allocation objectives and real-time computational scalability.



\bibliographystyle{IEEEtran}
\bibliography{BIB}

@article{attar2023data,
	title={Data-driven robust backward reachable sets for set-theoretic model predictive control},
	author={Attar, Mehran and Lucia, Walter},
	journal={IEEE Control Systems Letters},
	volume={7},
	pages={2305--2310},
	year={2023},
	publisher={IEEE}
}

@ARTICLE{Aslam2024Optimal,
  author={Aslam, Waleed and Andrianesis, Panagiotis and Caramanis, Michael},
  journal={IEEE Transactions on Sustainable Energy}, 
  title={Optimal {HVAC} Energy and Regulation Reserve Scheduling in Power Markets}, 
  year={2024},
  volume={15},
  number={1},
  pages={201-214},
}

@article{al2024efficient,
	title={An efficient method for quantifying the aggregate flexibility of plug-in electric vehicle populations},
	author={Al Taha, Feras and Vincent, Tyrone and Bitar, Eilyan},
	journal={IEEE Transactions on Smart Grid},
	volume={16},
	number={4},
	pages={3142--3154},
	year={2025},
	publisher={IEEE}
}

@article{barot2017concise,
	title={A concise, approximate representation of a collection of loads described by polytopes},
	author={Barot, Suhail and Taylor, Josh A},
	journal={International Journal of Electrical Power \& Energy Systems},
	volume={84},
	pages={55--63},
	year={2017},
	publisher={Elsevier}
}

@article{cordova2023aggregate,
	title={Aggregate modeling of thermostatically controlled loads for microgrid energy management systems},
	author={C{\'o}rdova, Samuel and Ca{\~n}izares, Claudio A and Lorca, {\'A}lvaro and Olivares, Daniel E},
	journal={IEEE Transactions on Smart Grid},
	volume={14},
	number={6},
	pages={4169--4181},
	year={2023},
	publisher={IEEE}
}

@article{chen2024unlock,
	title={Unlock the aggregated flexibility of electricity-hydrogen integrated virtual power plant for peak-regulation},
	author={Chen, Siqi and Zhang, Kuan and Liu, Nian and Xie, Yawen},
	journal={Applied Energy},
	volume={360},
	pages={122747},
	year={2024},
	publisher={Elsevier}
}

@article{chen2021leveraging,
	title={Leveraging two-stage adaptive robust optimization for power flexibility aggregation},
	author={Chen, Xin and Li, Na},
	journal={IEEE Transactions on Smart Grid},
	volume={12},
	number={5},
	pages={3954--3965},
	year={2021},
	publisher={IEEE}
}

@article{chen2024wholesale,
	title={Wholesale market participation of {DERAs: DSO-DERA-ISO} coordination},
	author={Chen, Cong and Bose, Subhonmesh and Mount, Timothy D and Tong, Lang},
	journal={IEEE Transactions on Power Systems},
	volume={39},
	number={5},
	pages={6605--6614},
	year={2024},
	publisher={IEEE}
}

@article{cui2024decision,
	title={Decision-oriented modeling of thermal dynamics within buildings},
	author={Cui, Xueyuan and Toubeau, Jean-Francois and Vallee, Francois and Wang, Yi},
	journal={IEEE Transactions on Smart Grid},
	volume={16},
	number={1},
	pages={369--382},
	year={2024},
	publisher={IEEE}
}

@article{han2024analytical,
	title={Analytical Solutions of air-conditioning load flexibility},
	author={Han, Binglong and Li, Hangxin and Wang, Shengwei},
	journal={IEEE Transactions on Smart Grid},
	volume={16},
	number={1},
	pages={441--444},
	year={2024},
	publisher={IEEE}
}

@article{hreinsson2021new,
	title={New insights from the {Shapley-Folkman} lemma on dispatchable demand in energy markets},
	author={Hreinsson, Kari and Scaglione, Anna and Alizadeh, Mahnoosh and Chen, Yonghong},
	journal={IEEE Transactions on Power Systems},
	volume={36},
	number={5},
	pages={4028--4041},
	year={2021},
	publisher={IEEE}
}

@article{hao2025non,
	title={Non-iterative coordinated optimal dispatch of distribution system and microgrids via fast approximation equivalent projection},
	author={Hao, Peng and Huang, Chunyi and Wang, Chengmin and Li, Kangping and Tan, Zhenfei},
	journal={IEEE Transactions on Smart Grid},
	volume={16},
	number={5},
	pages={3934--3948},
	year={2025},
	publisher={IEEE}
}

@article{he2025aggregative,
  title={Aggregative game analysis for energy consumption of multizone {HVAC} systems},
  author={He, Xiongnan and Lin, Zongli and Chang, Qing},
  journal={IEEE Transactions on Control of Network Systems},
  volume={12},
  number={3},
  pages={2194--2206},
  year={2025},
  publisher={IEEE}
}

@article{kurzhanskiy2011reach,
	title={Reach set computation and control synthesis for discrete-time dynamical systems with disturbances},
	author={Kurzhanskiy, Alex A and Varaiya, Pravin},
	journal={Automatica},
	volume={47},
	number={7},
	pages={1414--1426},
	year={2011},
	publisher={Elsevier}
}

@article{lorca2016multistage,
	title={Multistage adaptive robust optimization for the unit commitment problem},
	author={Lorca, Alvaro and Sun, X Andy and Litvinov, Eugene and Zheng, Tongxin},
	journal={Operations Research},
	volume={64},
	number={1},
	pages={32--51},
	year={2016},
	publisher={INFORMS}
}

@article{liu2024continuous,
	title={Continuous-time aggregation of massive flexible {HVAC} loads considering uncertainty for reserve provision in power system dispatch},
	author={Liu, Jingguan and Ai, Xiaomeng and Fang, Jiakun and Cui, Shichang and Wang, Shengshi and Yao, Wei and Wen, Jinyu},
	journal={IEEE transactions on Smart Grid},
	volume={15},
	number={5},
	pages={4835--4849},
	year={2024},
	publisher={IEEE}
}

@article{liu2025coupling,
	title={Coupling-Aware Aggregation of Multi-Zone {HVAC} Loads under Uncertainty: A Two-level Framework},
	author={Liu, Jingguan and Jiang, Han and Ai, Xiaomeng and Wang, Shengshi and Xue, Xizhen and Cui, Shichang and Hou, Jinming and Fang, Jiakun and Wen, Jinyu},
	journal={IEEE Transactions on Smart Grid},
	  year={2026},
  volume={17},
  number={3},
  pages={2077-2091},
	publisher={IEEE}
}

@article{liu2025leveraging,
	title={Leveraging Time-Causal State Variable Aggregation for Real-Time Schedule of Massive Air Conditioners},
	author={Liu, Jingguan and Ai, Xiaomeng and Cui, Shichang and Xue, Xizhen and Wang, Shengshi and Fang, Jiakun and Wen, Jinyu and Shi, Yang},
	journal={IEEE Transactions on Smart Grid},
	volume={16},
	number={3},
	pages={2389-2403},
	year={2025},
	publisher={IEEE}
}

@article{liu2026scalable,
	title={Scalable Building {HVAC} Control Through Laxity-Based Reinforcement Learning},
	author={Liu, Ruohong and Pan, Yuxin and Chen, Yize},
	journal={IEEE Transactions on Smart Grid},
	  year={2026},
  volume={17},
  number={5},
  pages={3898-3909},
	publisher={IEEE}
}

@article{rousseau2025uncertainty,
	title={Uncertainty-Aware Flexibility of {HVAC} Systems in Buildings: From Quantification to Provision},
	author={Rousseau, Julie and Cai, Hanmin and Heer, Philipp and Orehounig, Kristina and Hug, Gabriela},
	journal={IEEE Transactions on Smart Grid},
	year={2026},
	volume={17},
	number={1},
	pages={269-282},
	publisher={IEEE}
}

@article{ryzhov2019model,
	title={Model predictive control of indoor microclimate: existing building stock comfort improvement},
	author={Ryzhov, A and Ouerdane, Henni and Gryazina, Elena and Bischi, Aldo and Turitsyn, K},
	journal={Energy Conversion and Management},
	volume={179},
	pages={219--228},
	year={2019},
	publisher={Elsevier}
}

@article{song2018thermal,
	title={Thermal battery modeling of inverter air conditioning for demand response},
	author={Song, Meng and Gao, Ciwei and Yan, Huaguang and Yang, Jianlin},
	journal={IEEE Transactions on Smart Grid},
	volume={9},
	number={6},
	pages={5522--5534},
	year={2018},
	publisher={IEEE}
}

@article{tian2022real,
	title={Real-time flexibility quantification of a building {HVAC} system for peak demand reduction},
	author={Tian, Guanyu and Sun, Qun Zhou and Wang, Wenyi},
	journal={IEEE Transactions on Power Systems},
	volume={37},
	number={5},
	pages={3862--3874},
	year={2022},
	publisher={IEEE}
}

@article{tan2023non,
	title={Non-iterative solution for coordinated optimal dispatch via equivalent projection—{Part I: Theory}},
	author={Tan, Zhenfei and Yan, Zheng and Zhong, Haiwang and Xia, Qing},
	journal={IEEE Transactions on Power Systems},
	volume={39},
	number={1},
	pages={890--898},
	year={2024},
	publisher={IEEE}
}

@article{wang2025non,
	title={Non-iterative coordination of interconnected power grids via dimension-decomposition-based flexibility aggregation},
	author={Wang, Siyuan and Feng, Cheng and You, Fengqi},
	journal={IEEE Transactions on Power Systems},
	  year={2026},
	volume={41},
	number={1},
	pages={705-717},
	publisher={IEEE}
}

@article{wetzlinger2025backward,
	title={Backward reachability analysis of perturbed continuous-time linear systems using set propagation},
	author={Wetzlinger, Mark and Althoff, Matthias},
	journal={IEEE Transactions on Automatic Control},
	 year={2026},
  volume={71},
  number={1},
  pages={352-367},
	publisher={IEEE}
}

@article{wu2026energy,
	title={Energy Efficiency of Commercial {HVAC}-based Virtual Batteries for Load Shifting},
	author={Wu, Weimin and Lei, Shunbo and Sun, Qun Zhou and Mathieu, Johanna L},
	journal={IEEE Transactions on Power Systems},
	 year={2026},
  volume={41},
  number={1},
  pages={454-467},
	publisher={IEEE}
}

@article{wang2023control,
	title={A control framework to enable a commercial building {HVAC} system for energy and regulation market signal tracking},
	author={Wang, Wenyi and Tian, Guanyu and Sun, Qun Zhou and Liu, Hongrui},
	journal={IEEE Transactions on Power Systems},
	volume={38},
	number={1},
	pages={290--301},
	year={2023},
	publisher={IEEE}
}

@article{xu2025frequency,
	title={Frequency regulation provision by {EV} aggregator with dynamic balance between user charging anxiety and aggregate flexibility},
	author={Xu, Biao and Luan, Wenpeng and Zhao, Bochao and Long, Chao},
	journal={IEEE Transactions on Smart Grid},
	volume={16},
	number={4},
	pages={2953--2967},
	year={2025},
	publisher={IEEE}
}

@ARTICLE{yang2020HVAC,
  author={Yang, Yu and Hu, Guoqiang and Spanos, Costas J.},
  journal={IEEE Transactions on Automation Science and Engineering}, 
  title={{HVAC} Energy Cost Optimization for a Multizone Building via a Decentralized Approach}, 
  year={2020},
  volume={17},
  number={4},
  pages={1950-1960},
}

@article{ye2025aggregation,
	title={Aggregation of heterogeneous {DERs} by the subgradient vertex search method},
	author={Ye, Wenkai and Xie, Yilin and Xiao, Yao and Wei, Yifan and Chen, Jiali and Li, Yalun and Wang, Hewu and Ouyang, Minggao},
	journal={IEEE Transactions on Smart Grid},
	 year={2025},
	volume={16},
	number={6},
	pages={5053-5064},
	publisher={IEEE}
}

@article{yu2025flexibility,
	title={Flexibility aggregation and power disaggregation of distributed shared energy storage for coordinating virtual power plant in regulation market},
	author={Yu, Hongfei and Yu, Shunjiang and Xu, Fucong and Chen, Yuge and Zhang, Tianhan and Li, Zhicheng and Zhang, Weijun and Yang, Li and Lin, Zhenzhi},
	journal={IEEE Transactions on Sustainable Energy},
	  year={2026},
	volume={17},
	number={2},
	pages={970-984},
	publisher={IEEE}
}

@article{yan2023real,
	title={Real-time feedback based online aggregate {EV} power flexibility characterization},
	author={Yan, Dongxiang and Huang, Shihan and Chen, Yue},
	journal={IEEE Transactions on Sustainable Energy},
	volume={15},
	number={1},
	pages={658--673},
	year={2024},
	publisher={IEEE}
}

@article{zhao2017geometric,
	title={A geometric approach to aggregate flexibility modeling of thermostatically controlled loads},
	author={Zhao, Lin and Zhang, Wei and Hao, He and Kalsi, Karanjit},
	journal={IEEE Transactions on Power Systems},
	volume={32},
	number={6},
	pages={4721--4731},
	year={2017},
	publisher={IEEE}
}

@article{zhao2025analytical,
	title={An analytical feasibility condition for the multi-stage robust scheduling of energy storage systems with application on {SCUC}},
	author={Zhao, Jiexing and Zhai, Qiaozhu and Zhou, Yuzhou and Cao, Xiaoyu},
	journal={IEEE Transactions on Power Systems},
	volume={40},
	number={1},
	pages={435--448},
	year={2025},
	publisher={IEEE}
}

@article{zheng2024capacity,
	title={Capacity aggregation and online control of clustered energy storage units},
	author={Zheng, Boshen and Wei, Wei and Xu, Yin and Chen, Yue},
	journal={IEEE Transactions on Sustainable Energy},
	volume={15},
	number={3},
	pages={1546--1561},
	year={2024},
	publisher={IEEE}
}

@inproceedings{sadraddini2019linear,
	title={Linear encodings for polytope containment problems},
	author={Sadraddini, Sadra and Tedrake, Russ},
	booktitle={2019 IEEE 58th conference on decision and control (CDC)},
	pages={4367--4372},
	year={2019},
	organization={IEEE}
}

@inproceedings{yin2019finite,
	title={Finite horizon backward reachability analysis and control synthesis for uncertain nonlinear systems},
	author={Yin, He and Packard, Andrew and Arcak, Murat and Seiler, Peter},
	booktitle={2019 American Control Conference (ACC)},
	pages={5020--5026},
	year={2019},
	organization={IEEE}
}

@misc{FERC_Order2222,
	title        = {{FERC} Order No. 2222: Fact Sheet},
	howpublished = {[Online]. Available: \url{https://www.ferc.gov/media/ferc-order-no-2222-fact-sheet}},
	note         = {Accessed: Nov.~6, 2025},
	year         = {2020},
	month        = {Sep.}
}

@misc{gurobi, 
	author = {{Gurobi Optimization, LLC}}, 
	title = {{Gurobi Optimizer Reference Manual}}, 
	year = 2026, 
	howpublished = {[Online]. Available:
    \url{https://www.gurobi.com}}, 
    note         = {Accessed: Aug.~6, 2026}
}

\end{document}